\documentclass{article}                   
\usepackage{mathptmx} 
\usepackage{latexsym}
\usepackage{graphicx}
\usepackage{color}
\usepackage{amsmath}
\usepackage{amsthm}
\usepackage{amsfonts}
\usepackage{amssymb}
\usepackage{url}
\usepackage{wrapfig}
\usepackage{mathtools}
\mathtoolsset{showonlyrefs}
\usepackage{paralist}
\usepackage{microtype} 
\usepackage[utf8]{inputenc}
\usepackage[T1]{fontenc}
\usepackage{newtxtext}
\usepackage[vvarbb]{newtxmath}
\usepackage[numbers]{natbib}
\usepackage{enumerate}

\definecolor{darkblue}{RGB}{0, 0, 102}
\definecolor{IKB}{RGB}{0, 47, 167}

\usepackage[pdftex,bookmarks,bookmarksopen=true,pdfstartview={FitH},bookmarksnumbered=true,colorlinks=true,citecolor=darkblue,linkcolor=IKB,plainpages=false,pdfpagelabels]{hyperref}

\newtheorem{theorem}{Theorem}[section] 
\newtheorem{maintheorem}[theorem]{Main Theorem}
\newtheorem{lemma}[theorem]{Lemma}
 
\newtheorem{cor}[theorem]{Corollary}

\newtheorem{Def}[theorem]{Definition}

\theoremstyle{remark}
\newtheorem{remark}[theorem]{Remark}
\newcommand{\R}{\mathbb{R}}
\newcommand{\conv}{\mathop{\rm conv}}

\usepackage{authblk}
\begin{document}
\title{New Limits of Treewidth-based Tractability in Optimization}
\author[1]{Yuri Faenza}
\author[2]{Gonzalo Mu\~noz}
\author[3]{Sebastian Pokutta}
\affil[1]{Industrial Engineering and Operations Research, Columbia University}
\affil[2]{IVADO Fellow,  Canada Excellence Research Chair in Data
Science for Real-Time Decision-Making}
\affil[3]{Industrial and Systems Engineering, Georgia Institute of Technology}
\date{}

\maketitle

\begin{abstract}
Sparse structures are frequently sought when pursuing tractability in optimization problems.
They are exploited from both theoretical and computational perspectives to
handle complex problems that become manageable when sparsity is present.
An example of this type of structure is given by treewidth: a graph theoretical parameter that measures how ``tree-like'' a graph is. This parameter has been
used for decades for analyzing the complexity of various optimization problems and for
obtaining tractable algorithms for problems where this parameter is bounded. The goal of this work is to contribute to the understanding of the limits of the treewidth-based tractability in optimization. Our results are as follows. First, 
we prove that, in a certain sense, the already known positive results on extension complexity based on low treewidth are the best possible.
Secondly,
under mild assumptions, we prove that treewidth is the only
graph-theoretical parameter that yields tractability a wide class of optimization
problems, a fact well known in \emph{Graphical Models} in Machine Learning and in \emph{Constraint Satisfaction Problems}, which here we
extend to an approximation setting in \emph{Optimization}.
\end{abstract}

\section{Introduction}

Treewidth is a graph-theoretical parameter used to measure, roughly
speaking, how far a graph is from being a tree. It
was explicitly defined by \citet{Robertson198449} (also see
\cite{rs86}), but there are many equivalent definitions.  An earlier
discussion is found in \cite{halin} and closely related concepts have
been used by many authors under different names, e.g., the ``running
intersection'' property, and the notion of ``partial k-trees". Here we
make will use the following definition; recall that a \emph{chordal
  graph} is a graph where every induced cycle has exactly three
vertices.

\begin{Def} \label{treewidthdef} An undirected graph $G=(V,E)$ has
  \emph{treewidth} $\le \omega$ if there exists a chordal graph
  $H=(V,E')$ with $E\subseteq E'$ and clique number $\le \omega+1$. We
  denote as $tw(G)$ the treewidth of $G$.
\end{Def}

Note that $H$ in the definition above is sometimes referred to as a
\emph{chordal completion} of $G$. It can be shown that a graph has
treewidth \(1\) if and only if it is a forest. On the other extreme, a
complete graph of $n$ vertices has treewidth $n-1$. An important 
fact is that an $n$-vertex graph with treewidth $\leq \omega$ has
$O(\omega n)$ edges, and thus low treewidth graphs are \emph{sparse},
although the converse is not true. This sparsity is accompanied by a
compact decomposition of low-treewidth graphs that allows to
efficiently address various combinatorial problems.

Bounded treewidth has been long and widely recognized as a measure of
complexity for all kinds of problems involving graphs and there is
expansive literature concerning polynomial-time algorithms for
combinatorial problems on graphs with bounded treewidth.  One of the
earliest references is \cite{arnborg0}; see also \cite{arnborg,acp,BrownFellowsLangston1989,Bern1987216,bodlaender,bienstocklangston}. These algorithms typically rely on \emph{Dynamic
  Programming} techniques that yield algorithms with a non-polynomial
dependency on the treewidth. A similar paradigm has been presented in
\emph{Inference Problems} of Graphical Models (see, e.g.,
\cite{graphical}), where it is well known that an
underlying graph with bounded treewidth yields tractable inference
problems; see \cite{pearl82,freuder,dechterpearl,wainjor,Chandrasekaran08,wainwrightjordanlong} and references therein.

In a more general optimization context, treewidth-based sparsity has
been studied using the concept of the \emph{intersection graph}\footnote{Also called \emph{primal constraint graph} or \emph{Gaifman graph.}},
which provides a representation of the variable interactions in a
system of constraints. The intersection graph of a system of
constraints was originally introduced in \cite{fulkersongross} and has
been used by many authors, sometimes using different terminology.

\begin{Def} \label{intersectgraph} The \emph{intersection graph} of a
  system of constraints is the undirected graph which has a vertex for
  each variable and an edge for each pair of variables that appear in
  any common constraint.  If an optimization problem instance or its
  system of constraints is denoted $\mathcal{I}$, we call its
  intersection graph $\Gamma[\mathcal{I}]$.
\end{Def}

As it has been observed before (see \cite{treewidth,tw,Laurentima,lasserre2,Waki06sumsof,wainjor}), the combination of
\emph{intersection graph} and \emph{treewidth} makes it possible to
define a notion of structured sparsity in an optimization context. One
example of a research stream that has made use of treewidth-based
sparsity via intersection graphs is that of constraint satisfaction
problems (CSPs). One can obtain efficient algorithms for CSPs,
whenever the intersection graph of the constraints exhibits low
treewidth. Moreover, one can find compact linear extended formulations
(i.e., linear formulations with a polynomial number of constraints) in such
cases \cite{kolman2015extended,kolman2015extension}. In the Integer
Programming context, extended formulations for binary problems whose
constraints present a sparsity pattern with small treewidth have been
developed as well; see \cite{tw,wainjor,Laurentima}. A different use
of treewidth in Integer Programming is given in \cite{cunngeel}.  An
alternative perspective on structured sparsity in optimization
problems, without relying on an intersection graph, is taken in
\cite{BraunPokutta}.

\subsection*{Contribution}
\label{sec:contribution}

In this article we focus on two questions related to tractability
induced by treewidth. While we provide a precise statements of these
questions in each corresponding section, roughly speaking
these questions and our contribution can be summarized as follows:

\begin{enumerate}
\item In general, whenever an optimization problem exhibits an
  intersection graph with bounded treewidth, it can be solved (or
  approximated, depending on the nature of the problem) in polynomial
  time (see \cite{treewidth,kolman2015extension,Laurentima}). As
  such it is natural to ask the following question:
\begin{center}
\emph{Is there any other graph-theoretical structure that yields tractability?}
\end{center}

It is known that the answer is negative in general. \citet{grohe2007complexity} and \citet{marx2010can} proved that, in a sense, CSPs are only tractable when bounded treewidth is present. \citet{Chandrasekaran08} proved that a family of graphs with unbounded treewidth can yield intractable inference problems in Graphical Models, under the $NP\not\subseteq P/poly$ hypothesis. 
Moreover, it is believed that many treewidth-based algorithms are best possible \cite{lokshtanov2011known}.

We complement these results by proving that a family of graphs with unbounded treewidth can yield intractable optimization problems, even if the variable domain is bounded and small violations to the constraints are allowed. This provides a converse to a recent theorem by \citet{treewidth}. We follow the overall strategy of \citeauthor{Chandrasekaran08}, but we make use of the hypothesis $NP\not\subseteq BPP$ instead. Besides the different complexity theory assumptions, we highlight other differences of our approach and results compared to that of \citet{Chandrasekaran08} and \citet{marx2010can} in Sections \ref{sec:negativeresultsCSP} and \ref{sec:chandrasekaran}.

\item For sets in $\{0,1\}^n$ defined using a set of constraints whose
  intersection graph has treewidth $\omega$, it is known that there
  exists a linear programming reformulation of its convex hull of size
  $O(n2^\omega)$. This yields the following question:
\begin{center}
\emph{For any given treewidth $\omega$, is there any 0/1 set that (nearly) meets this bound?}
\end{center}
We provide a positive answer to this question. Furthermore, we prove that this
bound is tight even if we allow semidefinite programming
formulations. This establishes that there is little to be gained from
semidefinite programs over linear programs in general when exploiting low
treewidth. Our analysis is based on
the result of \citet{briet2013existence}, where the existence of 0/1 sets with exponential
semidefinite extension complexity is proved. We also prove a similar result for the stable set polytope, making use of the treewidth of the underlying graph directly instead of relying on a particular formulation, and discuss related results.

It is worth mentioning
that the extension complexity upper bound is obtained enumerating
\emph{locally} feasible vectors along with a gluing
argument. Moreover, the upper bound is oblivious to any other
structure present in the constraints besides its sparsity pattern. Our
result shows that, surprisingly, one cannot do much better than this
seemingly straight-forward approach, even if semidefinite formulations
are allowed. 

Typically, one can find treewidth-based \emph{upper bounds} on the
extension complexity of certain polytopes \cite{kolman2015extension,kolman2015extended,BraunPokutta}, or extension complexity
\emph{lower bounds} on specific families of problems
\cite{braun2015approximation,BraunPokutta,fiorini2015exponential,avis2015extension} parametrized
using the problem size. To the best of our knowledge, much less attention has been devoted to providing extension complexity lower bounds parameterized using other features of the problem. As a matter of fact, we are only aware of two other articles in this domain: the work of \citet{gajarsky2017parameterized}, where the authors analyze the
extension complexity of the stable set polytope based on the
\emph{expansion} of the underlying graph, and the work of \citet{aboulker2019extension} which, independently of this work, provided extension complexity lower bounds of the correlation polytope parameterized using the treewidth of the underlying graph. Our work contributes to this
line of work, showing the existence of polytopes whose extension complexity lower bound depends on the treewidth parameter and nearly meet the aforementioned bound. We discuss the main difference of our approach to that of \citeauthor{aboulker2019extension} in Section \ref{sec:lowtreewidth}.
\end{enumerate}

We believe that addressing these two questions provides new valuable
insights into the limitations of exploiting treewidth and provides
strong lower bounds that allow for assessing the performance of
current approaches. In fact, complementing the results by \citet{Chandrasekaran08} and \citet{marx2010can}, our results show that the existing
approaches are, in some sense, the best possible and that further
improvement is only possible if more structure is considered.

We emphasize that the two questions studied in this paper, although both related to treewidth, are different and their answers need distinct approaches and tools. The extension complexity is a concept that does not
necessarily depend on whether a problem is easy or hard from an
algorithmic perspective, nor on the assumption of $P\neq NP$. For
example, there are instances of the matching polytope with an
exponential extension complexity \cite{rothvoss2017matching}, whereas
finding a maximum weight matching can be done in polynomial time for
any graph. An example in the other direction is given by the stable
set problem. For each $\epsilon>0$, an $n^{1-\epsilon}$-approximate
solution cannot be attained in polynomial time
\cite{Hastad01,Zuckerman07} unless $P=NP$, but there exists a
formulation of polynomial size of the stable set polytope with the
property that, for each objective function $c\geq 0$, its optimal
solution is a factor $O(\sqrt{n})$ away from the maximum weight stable
set (\cite{bazzi2015no,bazzi2015noJour}, by building on results from
\cite{Feige}).

\subsection*{Outline}
\label{Outline}
The rest of the article is organized as follows. In Section
\ref{notation} we provide the basic notation used in this article. The
main contributions are divided in two sections. In Section
\ref{unboundedTW} we provide the answer to the first question above,
i.e., we prove that unbounded treewidth can yield intractable
optimization problems, even if constraint violations are allowed, and in Section \ref{XClowerbounds} we provide
the answer to the second question, i.e., we show the existence of
sparse problems with high extension complexity. Both sections are
organized similarly: we begin by providing the necessary background,
along with the known positive treewidth-exploiting results, and then
move to the respective proofs. Section \ref{extra-results} provides
additional results to complement Section
\ref{XClowerbounds}.

\section{Notation} \label{notation}

We mostly follow standard linear algebra and graph theory
notation. For $n\in \mathbb{N}$, we use $[n]$ to denote the set of
integers $\{1,\ldots, n\}$. Further, we denote by $\mathbb{R}^n$ the
\(n\)-dimensional vector space of the reals and by $\mathbb{Z}^n$ the
$n$-dimensional free \(\mathbb{Z}\)-module over the integers.  If we
restrict vectors to have non-negative entries, we use $\mathbb{R}^n_+$
and $\mathbb{Z}^n_+$. We call $e_i$ with \(i \in [n]\) the canonical
vectors in $\mathbb{R}^n$, i.e., $(e_i)_j = 1$ if and only if
$i=j$. The space of symmetric $n\times n$ positive semidefinite
matrices is denoted as $\mathbb{S}^n_+$. The standard inner product
between two vectors $v,w\in \mathbb{R}^n$ is denoted by $v^T w$. Given
two matrices \(A,B\) (of compatible dimension), the Frobenius inner
product is denoted by
\(\langle A, B\rangle \doteq \text{trace}(A^T B)\). Given two set
$S_1, S_2$, we denote the cartesian product by
$S_1 \times S_2 \doteq \{(x_1, x_2) \ : \ x_1 \in S_1,\, x_2 \in
S_2\}$. The convex hull of a set $S\subseteq \mathbb{R}^n$ is denoted
as $\conv(S)$, and its affine hull by $\text{aff}(S)$. For a graph $G=(V,E)$, we use $V(G)$ and $E(G)$ to denote its vertices and edges respectively. For $v\in V(G)$, we use $N_G(v)$ to denote the set of neighbors of $v$ in $G$, that is $N_G(v) = \{u \, : \, \{u,v\} \in E(G)\}$. Given two graphs
$G_i=(V_i,E_i)$ with $i \in \{1,2\}$, we have that $G_1$ is a subgraph
of $G_2$ if $V_1\subseteq V_2$ and $E_1 \subseteq E_2$, and $G_1$ is a
minor of $G_2$ if $G_1$ can be obtained from $G_2$ using \emph{vertex
  deletions, edge deletions, and edge contractions}. Lastly, for a
polynomial $p(x)$, we denote by $\|p\|_1$ the sum of the absolute
values of its coefficients, i.e., if
\(p(x) = \sum_{\alpha \in I(p)} p_{\alpha} x^\alpha\) with
$x^\alpha \ \doteq \ \prod_{j=1}^n x_j^{\alpha_j}$ for some
$\alpha \in \mathbb{Z}_+^n$, $p_{\alpha}$ rational and
$I(p) \subseteq \mathbb{Z}_+^n$, then
\[\|p \|_1 \ \doteq \ \sum_{\alpha \in I(p)} |p_{\alpha}|.\]
The degree of $p$ is defined as  \(\deg(p) \doteq \max_{\alpha \in I(p)} \sum_j \alpha_j\).


\section{Unbounded treewidth can yield intractability} \label{unboundedTW}

Our first goal is to study the question of whether low treewidth is
the \emph{only} graph-theoretical structure that yields tractability
when approximating optimization problems. Here we work with the general \emph{Polynomial
  Optimization} framework, i.e., we consider problems of the form:
\begin{subequations} \label{PO}
\begin{align}
\mbox{(PO): }  \min\, &\, c^T x \\
\mbox{s.t. }\, & f_i(x) \ \ge \ 0 && i\in [m]\\
& x_j \in \{0,1\} && j\in [p],\\
& x_j \in [0,1] && j = p+1,\ldots, n. 
\end{align}
\end{subequations}
where each $f_i$ is a polynomial of degree at most $\rho$. When $\rho=2$ we also use the term QCQP (quadratically constrained quadratic problem) to refer to PO.  
\begin{remark}
Any problem with polynomial objective and constraints, and defined over a compact
set, can be cast as a PO. This can be done by appropriately rescaling variables and by using an epigraph formulation to move the non-linear terms of the objective to the constraints.
\end{remark}

As mentioned above, it is known that tractability of an instance
$\mathcal{I}$ of PO is implied by an intersection graph
$\Gamma[\mathcal{I}]$ of low treewidth. In the pure binary case, an
\emph{exact} optimal solution of $\mathcal{I}$ can be computed in
polynomial time whenever $\Gamma[\mathcal{I}]$ has bounded treewidth
(see \cite{treewidth,kolman2015extension,Laurentima}). However, if
continuous variables are present, exact solutions might not be
computable in finite time as shown by the following simple example.
\begin{align*}
\max \, &\, x \\
\mbox{s.t. }\, &\, x^2 \ \le \ \frac{1}{2}  \\
& x\in [0,1]
\end{align*}
has an irrational optimal solution. As such approximation is
unavoidable from a computational perspective, therefore we make use of
the following definition:

\begin{Def}
Given an instance $\mathcal{I}$ of PO, we say $x^* \in \{0,1\}^p \times [0,1]^{n-p}$ is $\epsilon$-feasible if $x^*\in S_\epsilon$, where
\begin{equation}
    \label{def:Sepsilon}
    S_\epsilon = \{x\in \{0,1\}^p \times [0,1]^{n-p} \ : \ f_i(x) \geq -\epsilon \|f_i\|_1, \ 1\leq i \leq m\}.
\end{equation}
\end{Def}

Given an instance $\mathcal{I}$ of PO an LP formulation that takes
advantage of low treewidth of $\Gamma[\mathcal{I}]$ was proposed by
\citet{treewidth} in order to approximate $\mathcal{I}$. More
specifically:

\begin{theorem}[Bienstock and Mu\~noz \cite{treewidth}] \label{potheorem_0} Consider a
  feasible instance $\mathcal{I}$ of PO and $\epsilon > 0$. Assume each $f_i(x)$ has
  degree at most $\rho$. If $\Gamma[\mathcal{I}]$ has treewidth
  $\le \omega$ then there is an LP formulation with
  $O\left((2 \rho/\epsilon)^{\omega + 1} \, n \log
    (\rho/\epsilon)\right)$ variables and constraints such that
\begin{enumerate}[(a)]
\item all feasible solutions to the LP are $\epsilon$-feasible for $\mathcal{I}$
\item every optimal LP solution $\hat{x}$ satisfies
\begin{equation}\label{eps-opt}
c^T \hat{x} \leq c^T x^* + \epsilon \| c_N \|_1
\end{equation}
where $x^*$ is an optimal solution to $\mathcal{I}$ and $c_N$ is the
sub-vector of $c$ corresponding to continuous variables $j=p+1,\ldots, n$.
\end{enumerate}

Moreover, given a chordal completion of $\Gamma[\mathcal{I}]$ with clique number $\leq \omega + 1$ (which exists whenever the treewidth is at most \(\omega\)), the LP
can be constructed in time
\[O( (2 \rho/\epsilon)^{\omega + 1} \, \log (\rho/\epsilon) \cdot \textit{poly}(\|\mathcal{I}\|)).\]
where $\|\mathcal{I}\|$ is the size of the representation of $\mathcal{I}$.
\end{theorem}

Here we phrased the theorem in a slightly different way compared to
\cite{treewidth}:
\begin{inparaenum}[(a)]
  \item we assume that $\mathcal{I}$ is feasible and
  \item the result in \cite{treewidth} \emph{only} considers
    continuous variables, whereas we allow for binary variables as
    well. This can be done while ensuring that the error term in
    \eqref{eps-opt} only involves coefficients associated with
    continuous variables; see \cite{thesisGonzalo} for details.
\end{inparaenum}

We would like to stress that the approximation provided by
Theorem \ref{potheorem_0} is different from the traditional notion of
approximation used in approximation algorithms: we allow for
$\epsilon$-feasibility, i.e., we allow (slightly) infeasible
solutions, which is usually not the case in approximation algorithms.

For $\rho = 2$ we obtain the following immediate corollary of Theorem
\ref{potheorem_0}. 

\begin{cor} \label{cor:0} For every fixed $\epsilon>0$, 
   there is an algorithm $\mathcal{A}$ such that, given a feasible instance
   $\mathcal{I}$ of QCQP and a chordal completion of $\Gamma[\mathcal{I}]$ with clique number $\leq \omega + 1$,
   it computes an $\epsilon$-feasible solution
   satisfying \eqref{eps-opt} in time $O(C^\omega \textit{poly}(\| \mathcal{I} \|))$, where $C$ is a constant.
 \end{cor}

We establish an (almost) matching lower bound for Theorem
\ref{potheorem_0} by providing an (almost) matching lower bound for
Corollary \ref{cor:0}. For this we use the strategy of
\citet{Chandrasekaran08} adapted to the general optimization
setting. We make use of the following definition:

\begin{Def}
We say a countable family of graphs $\{\mathcal{G}_k\}_{k=1}^\infty$ is \emph{polynomial-time enumerable} if there is an algorithm such that, given $k$, it outputs a description of $\mathcal{G}_k$ in time $\textit{poly}(k)$.
\end{Def}

Using this definition, we prove the following; we discuss the
complexity theoretic assumption $NP \not\subseteq BPP$ in Section \ref{tools}:

\begin{maintheorem}\label{main-theo}
  Fix $\epsilon < 1/10$ and let $\{\mathcal{G}_k\}_{k=1}^\infty$ be an arbitrary polynomial-time enumerable
  family of graphs indexed by treewidth. Let
  $\mathcal{A}$ be an algorithm such that for all instances
  $\mathcal{I}_k$ of QCQP such that
  $\Gamma[\mathcal{I}_k] = \mathcal{G}_k$ algorithm \(\mathcal{A}\)
  computes an $\epsilon$-feasible solution satisfying \eqref{eps-opt} in time $T(k)\cdot \textit{poly}(\|\mathcal{I}_k \|)$, then assuming 
  $NP \not\subseteq BPP$ implies that $T(k)$ grows super-polynomially in $k$.
\end{maintheorem}

Note that assuming the family is polynomial-time enumerable implies that an encoding of $\mathcal{G}_k$ of size polynomial in $k$ exists. This is indeed a desirable feature, since we will be dealing with a polynomial-time reduction, and thus we need to have at least an efficient access to the graph family. In fact, \citet{Chandrasekaran08} assume this implicitly, as they assume access to the graph family via a polynomial time ``advice''.

\subsection{Related intractability results in CSPs}\label{sec:negativeresultsCSP}
Many treewidth-based intractability results have been obtained in the CSP community. Two crucial contributions are those of \citet{grohe2007complexity} and \citet{marx2010can} who proved that treewidth, in a sense, is the only tractable graph structure in a CSP. More specifically, assuming FPT$\neq$W[1], \citet{grohe2007complexity} proved that CSPs defined over a recursively enumerable family of graphs are polynomially solvable if and only if the family has bounded treewidth. Later on, \citet{marx2010can} proved the following result that, assuming stronger complexity theoretic assumptions, leads to sharper lower bounds.

\begin{theorem}{(\citet[Theorem 1.3]{marx2010can})}\label{thm:marx}
If there is a class $\mathcal{G}$ of graphs with unbounded treewidth, an algorithm $\mathcal{A}^M$, and a function $f$ such that $\mathcal{A}^M$ correctly decides every binary CSP instance and the running time is $f(G) \| I \|^{o(tw(G)/\log tw(G))}$
for binary $CSP(G)$ instances $I$ with intersection graph $G \in \mathcal{G}$, then the Exponential Time Hypothesis (ETH) fails.
\end{theorem}
Here \emph{binary CSP} refers to CSP problems where each constraint involves at most two variables, and does \emph{not} imply that the variables' domain is $\{0,1\}$. Note that $\mathcal{A}^M$ in Theorem \ref{thm:marx} is assumed to be defined over \emph{all} CSP instances, meaning, instances with \emph{any} intersection graph (although the running-time assumption is only made on the family $\mathcal{G}$). However, \citeauthor{marx2010can} also provides an alternative result that requires $\mathcal{A}^M$ to be defined only on CSPs whose intersection graph belongs to the family $\mathcal{G}$, under the assumption of $\mathcal{G}$ being recursively enumerable.

From this, it is natural to ask whether Theorem \ref{main-theo} can be obtained from these
already known results. We argue why this is not the case and that our result is rather complementary.\\

The first evident difference lies in the complexity-theoretic assumption. \citet{grohe2007complexity} assumes FPT$\neq$W[1], \citet{marx2010can} assumes ETH, whereas we assume $NP\not\subseteq BPP$.
$NP\not\subseteq BPP$, roughly speaking, asserts that certain problems in $NP$ cannot be solved in randomized polynomial time. Not much is known about the relationship of $NP$ and $BPP$, but is widely believed that $P = BPP$, which would make $NP\not\subseteq BPP$ equivalent to $P\neq NP$. We describe $BPP$ more precisely in Section \ref{tools}.  

Secondly, the results obtained by \citeauthor{grohe2007complexity} and \citeauthor{marx2010can} are impossibility results for solving CSPs \emph{exactly}, while our goal is to provide a converse to Theorem \ref{potheorem_0} ---an approximation-type of result.
Thus, we must allow algorithm $\mathcal{A}$ to return potentially infeasible solutions. It is not clear to us, and seems a challenging task, whether the sequence of reductions from e.g.~\cite{marx2010can} can be extended to prove bounds on the approximation guarantee.

Finally, and most importantly, \citeauthor{grohe2007complexity} and \citeauthor{marx2010can} deal with CSPs defined over unbounded domains. In this case, the treewidth-based algorithmic complexity upper bound is roughly $n^{O(\omega)}$, which is what the authors work with. In our case, algorithm $\mathcal{A}$ in Theorem \ref{main-theo} is only assumed to be defined over QCQP instances (which can be viewed as a subset of CSP instances) and whose variables' domain is only $\{0,1\}$ or $[0,1]$. This causes the upper bound in Corollary \ref{cor:0} to be better than $n^{O(\omega)}$ and thus we need a different 
procedure to provide an intractability result in this case. We note that due to the same observation, the result by \citet{Chandrasekaran08} that we discuss below does not follow from \citeauthor{marx2010can}'s.

\subsection{Intractability in the 0/1 case}\label{sec:chandrasekaran}

In the 0/1 case, a similar result to Theorem \ref{main-theo}
can be obtained as a direct consequence of the work of \citet{Chandrasekaran08} in the \emph{exact} setting, i.e., when no approximation is allowed. This was done
in the context of \emph{graphical models}.

Given a graph $G=(V,E)$, a collection of (binary)
random variables $x_v$ with $v\in V$, and for each $K\subseteq V$
forming a clique of $G$ a function $\psi_K$ which only involves
variables $x_v$ with $v\in K$, then the \emph{inference problem}
involves computing the partition function $Z(\psi)$ defined as
\[Z(\psi) = \sum_{x\in \{0,1\}^V} \prod_{K \in \mathcal{K}} \psi_{K}(x_K),\]
where $\mathcal{K}$ is the set of all cliques in $G$. It is known that
if the underlying graph $G$ has bounded treewidth, then the inference
problem can be solved in polynomial time
(see \cite{wainwrightjordanlong}). One of the main results in
\cite{Chandrasekaran08} provides a converse to this statement: given
any family of graphs $\{\mathcal{G}_k\}_{k=1}^\infty$ indexed by
treewidth---under the complexity assumptions of Theorem
\ref{main-theo}---there exist instances defined over that family of
graphs such that inference requires time super-polynomial in
$k$.

The proof can be directly adapted to state the same result
regarding computing an optimal solution for a 0/1 PO problem. Hence,
the result by \citet{Chandrasekaran08} can be viewed as the 0/1 version
Theorem \ref{main-theo}, which does not involve $\epsilon$-feasible
solutions, as in such context exact solutions can be computed.

\begin{remark}
  The original proof in \cite{Chandrasekaran08} makes use of the
  $NP\not\subseteq P/\textit{poly}$ hypothesis and the so called
  \emph{Grid-minor} hypothesis. Since then, the latter was shown to
  be true by \citet{chekuri2016polynomial}, along with an algorithmic
  result allowing the use of the $NP \not\subseteq BPP$ instead of
  $NP\not\subseteq P/\textit{poly}$.
\end{remark}

Here we extend these results to include continuous variables and show
that even approximately solving the problem remains intractable. The
proof is along the lines of \cite{Chandrasekaran08} and we follow
their overall strategy. Our contribution here is to replace reductions
between distributions and potential functions with reductions
involving QCQPs and their
approximations, as well as making use of randomized algorithms instead of non-uniform algorithms. To avoid confusion, we would like to stress that the
notion of \emph{Approximate Inference} presented in
\cite{Chandrasekaran08} is a different concept compared to finding an
$\epsilon$-feasible solution to a PO problem.

\subsection{Complexity-theoretic Assumptions and Graph-theoretic Tools}\label{tools}

For a precise definition of $BPP$ and the commonly believed 
$NP \not\subseteq BPP$ hypothesis,
we refer the reader to \cite{AB}. Simplifying here, $BPP$ is the class
of languages $L$ for which a polynomial time probabilistic Turing
machine exists which, given an input $x$, provides a wrong answer to the decision
$x\in L$ with probability of at most 1/3, whether in fact $x\in L$ or $x\not\in L$. 
In our context here, it is sufficient to know that this 
complexity-theoretic assumption implies that MAX-2SAT in planar graphs
(an NP-hard problem; see \cite{guibas1991approximating}) does not belong to
$BPP$.

The second important tool we will use stems from work on the famous graph minor
theorem. We briefly recall relevant results here, phrased to match the language in
\cite{Chandrasekaran08}.

\begin{theorem}[\citet{ROBERTSON1994323}] \label{minor1}
   There exist universal constants $c_3$ and
  $c_4$ such that the following holds. Let $G$ be a $g\times g$
  grid. Then, (a) $G$ is a minor of all planar graphs with treewidth
  greater than $c_3 g$. Further, (b) all planar graphs of size (number
  of vertices) less than $c_4 g$ are minors of $G$.
\end{theorem}

The next theorem relaxes the planarity assumption however
only for one of the directions of Theorem~\ref{minor1}.  

\begin{theorem}[\citet{ROBERTSON1994323}] \label{minor2}  Let $G$ be a
  $g\times g$ grid. There exists a finite $\kappa_{GM}(g)$ such that
  $G$ is a minor of all graphs with treewidth greater than
  $\kappa_{GM}(g)$. Further, $c_1 g^2 \log g \leq \kappa_{GM}(g) \leq 2^{c_2 g^5}$, where
  $c_1$ and $c_2$ are universal constants (i.e., they are independent
  of $g$).
\end{theorem}

The last theorem provides bounds on the magnitude of $\kappa_{GM}(g)$
in order for it to have the $g\times g$ grid as a minor. The constant
$\kappa_{GM}(g)$ was conjectured to be polynomial in $g$, and was used
as a complexity-theoretic assumption (under the name \emph{Grid-minor
  hypothesis}) in \cite{Chandrasekaran08}. Since then, a recent
breakthrough by \citet{chekuri2016polynomial} resolved this in the positive.

\begin{theorem}[\citet{chekuri2016polynomial}]\label{chekuri}

 \[\kappa_{GM}(g) \in O(g^{98}\textit{poly log}(g))\]
 Moreover, there is a polynomial time randomized algorithm that, given a graph G with treewidth at least $\kappa_{GM}(g)$, with high probability\footnote{probability at least $1 - 1/|V(G)|^c$ for some constant $c>1$} outputs the sequence of grid minor operations transforming $G$ into the $g \times g $ grid.
\end{theorem}

\begin{remark}
In \cite{chekuri2016polynomial}, the output of the randomized
algorithm is a \emph{model} of the minor. Such \emph{model} can be
directly turned into a set of minor operations.
\end{remark}

\begin{remark}
  There has been some considerable progress recently regarding the exponent of the polynomial dependency in Theorem \ref{chekuri}. We refer the reader to \cite{chuzhoy2016improved} for these improvements. Nonetheless, these newer results are non-algorithmic, which is undesirable for our purposes.
\end{remark}

Theorem \ref{chekuri}, together with Theorem~\ref{minor1}
and Theorem~\ref{minor2}, yields the following corollary:

\begin{cor} \label{main-cor}
Let $G = (V,E)$ be a planar graph of $n$ nodes. There exists a polynomial $\kappa (n)$ such that $G$ is a
minor of \emph{all} graphs of treewidth at least $\kappa (n)$. 
\end{cor}

The above in particular implies that $G$ is a minor of $\mathcal{G}_k$ for all $ k \geq
\kappa(n)$ for the sequence of graphs in Theorem~\ref{main-theo}. 

\subsection{Proof of Theorem
  \ref{main-theo}} \label{section2SATreduction} The outline of the
proof of Main Theorem \ref{main-theo} is as follows. We start from a
NP-hard instance $\mathcal{I}$ of QCQP, whose intersection graph
$\Gamma[\mathcal{I}]$ is planar. Recall that we assume we are given an
arbitrary family of graphs $\{\mathcal{G}_k\}_{k=1}^\infty$ indexed by
treewidth. Due to Corollary \ref{main-cor}, $\Gamma[\mathcal{I}]$ is a
minor of $\mathcal{G}_k$ for some $k$ large enough. We then construct
an instance $\mathcal{I}_{k}$ of QCQP equivalent to $\mathcal{I}$
whose intersection graph is exactly $\mathcal{G}_k$. This makes it
possible to use algorithm $\mathcal{A}$ over $\mathcal{I}_k$, which yields the
conclusion. The key ingredient is the following: having a family with
unbounded treewidth allows us to embed the graph defining the NP-Hard
problem into a graph of the given family, even if this family is
arbitrary.

\subsubsection{Formulating MAX-2SAT as a special PO
  problem} \label{first-step} Consider the NP-Hard problem of planar
MAX-2SAT with underlying planar graph $G=(V,E)$. Denote
$\{C_i\}_{i=1}^m$ the clauses and $E=\{e_i\}_{i=1}^m$ the edges of
\(G\). Let the variables be $x_j$ with $j \in [n]$. Then
\[e_i = \{x_{i_1}, x_{i_2}\} \Leftrightarrow C_i = \{x_{i_1} \vee x_{i_2}\} \vee  C_i = \{x_{i_1} \vee \overline{x_{i_2}}\} \vee  C_i = \{\overline{x_{i_1}} \vee x_{i_2}\} \vee  C_i = \{\overline{x_{i_1}} \vee \overline{x_{i_2}}\}\]

We can formulate MAX-2SAT directly as a QCQP:
\begin{subequations}
\begin{align}\mbox{(MAX-2SAT-1): } \max &\, \sum_{i=1}^m y_i \ \\
\mbox{s.t. } &  y_i - f_i(x_{i_1}, x_{i_2}) \ \le \ 0 && i\in [m]\\
&  x_j^2 - x_j = 0 && j\in [n]  \\
&  y_i \in \{0,1\} && i\in [m]  \\
&  x_j \in [0,1] &&  j\in [n],
\end{align}
\end{subequations}
where 
\[  f_i(x_{i_1}, x_{i_2}) = \left\{
\begin{array}{ll}
x_{i_1} + x_{i_2} & \text{if }  C_i = \{x_{i_1} \vee x_{i_2}\} \\
x_{i_1} + (1 - x_{i_2}) & \text{if }  C_i = \{x_{i_1} \vee \overline{x_{i_2}}\} \\
(1 - x_{i_1}) + x_{i_2} & \text{if }  C_i = \{\overline{x_{i_1}} \vee x_{i_2}\} \\
(1- x_{i_1}) + (1- x_{i_2}) & \text{if }  C_i = \{\overline{x_{i_1}} \vee \overline{x_{i_2}}\} \\
\end{array}\right.,\]
thus $y_i = 1$ implies that clause $C_i$ is satisfied. Let
$\mathcal{I}$ be an instance of MAX-2SAT-1. Note that using this formulation the graph \(G\) is a subgraph of the
intersection graph $\Gamma[\mathcal{I}]$. It is also not hard to see
that $\Gamma[\mathcal{I}]$ is planar, as we only need to add vertices
$y_i$, and each vertex $y_i$ is connected to the endpoints of one
particular edge of $G$. We would like to emphasize that constraints
$x_j^2 - x_j = 0 $ are equivalent to simply requiring
$x_j \in \{0,1\}$, so we could formulate MAX-2SAT as a pure binary
problem, however, since we are aiming for statements about the complexity of
approximating PO problems, we deliberately chose a formulation using variables that
can be continuous in nature; this will become clear soon. 

The above formulation of MAX-2SAT is straight-forward, however for
technical reasons we use the following equivalent alternative.  The
advantage of this formulation is that all constraints involve
only \(1\) or \(2\) variables, simplifying the later analysis.
\begin{subequations}
\begin{align}
\mbox{(MAX-2SAT): } \max\, & \sum_{i=1}^m (y_{i_1} + y_{i_2}) \ \\
\mbox{s.t. } & y_{i_1} - f_{i_1}(x_{i_1}) \ \le \ 0  && i\in[m]\\
&  y_{i_2} - f_{i_2}(x_{i_2})  \ \le \ 0 &&i\in [m] \\
&  y_{i_1} + y_{i_2}  \ \le \ 1 && i\in [m]\\
&  x_j^2 - x_j = 0 && j\in [n] \\
&  y_i \in \{0,1\} && i\in [m] \\
&  x_j \in [0,1] && j\in [n],
\end{align}
\end{subequations}
where
\[ f_{i_1}(x_{i_1}) = \left\{
\begin{array}{ll}
x_{i_1} & \text{if }  C_i = \{x_{i_1} \vee x_{i_2}\}  \text{ or }  C_i = \{x_{i_1} \vee \overline{x_{i_2}}\} \\
1- x_{i_1} & \text{if } C_i = \{\overline{x_{i_1}} \vee x_{i_2}\} \text{ or }  C_i = \{\overline{x_{i_1}} \vee \overline{x_{i_2}}\}
\end{array}\right.,\]
and $f_{i_2}(x_{i_2})$ is similarly defined.

\subsubsection{Graph Minor Operations} \label{minor-operations} Let
$\mathcal{I}$ be an instance of MAX-2SAT with planar intersection
graph $\Gamma[\mathcal{I}]$. Given a target graph $H$ which has $\Gamma[\mathcal{I}]$ as a minor, in this
section we show how to construct a QCQP instance $\mathcal{I}_H$
equivalent to $\mathcal{I}$. The complexity of this reduction is polynomial in the
number of minor operations (vertex deletion, edge deletion, edge
contraction), assuming that we know in advance which those operations
should be.  We will first show this for $H$ being
contractable to $\Gamma[\mathcal{I}]$ using a \emph{single} minor
operation and then argue that this is without loss of generality by
repeating the argument. We distinguish the following cases:

\begin{enumerate}[(a)]
\item \emph{Vertex Deletion.} If the minor operation is a vertex
  deletion of a vertex $u\in V(H)$, we define $\mathcal{I}_H$ as $\mathcal{I}$
  plus a new variable $x_u \in [0,1]$ with objective coefficient
  \(0\). Additionally, for all $ v\in N_{H}(u)$ we add the redundant
  constraint $x_v + x_u \geq 0$.
\item \emph{Edge Deletion.} If the minor operation is an edge deletion
  of an edge $(u,v) \in
  E(H)$, we define $\mathcal{I}_H$ as $\mathcal{I}$ plus the redundant
  constraint $x_v + x_u \geq 0$.
\item \emph{Edge contraction.} If the minor operation is an edge
  contraction of $(u,v) \in E(H)$ to form
  $w \in V(\Gamma[\mathcal{I}])$, then we proceed as follows.

  Let $N_H(u)$ be the neighbors of $u$ in $H$. Note that in
  $\mathcal{I}$ all constraints involve at most \(2\) variables, hence
  there is a one-to-one correspondence of edges in
  $\Gamma[\mathcal{I}]$ and constraints involving \(2\) variables in
  $\mathcal{I}$, and all these constraints are linear. Such
  constraints have the form
\[ a_{w,t} z_{w} + b_{w,t} z_{t} \leq d_{w,t} \quad t\in N_{\Gamma[\mathcal{I}]}(w),\]
where variables $z$ can be either variables $x$ or $y$ in MAX-2SAT,
depending on node $w$. Using this, we define $\mathcal{I}_H$ from
$\mathcal{I}$ by removing variable $z_{w}$, adding variables $z_{u}$
and $z_{v}$, and adding the following constraints
\begin{eqnarray*}
a_{w,t} z_{u} + b_{w,t} z_{t} &\leq& d_{w,t} \quad t\in N_{H}(u) \\ 
a_{w,t} z_{v} + b_{w,t} z_{t} &\leq& d_{w,t} \quad t\in N_{H}(v) \\ 
z_{u} &=& z_v
\end{eqnarray*}
If the objective value of $z_w$ was 1, then we ensure $z_{u}$ to have
objective value \(1\) and $z_v$ to have objective value \(0\). If $z_w$ was a
continuous variable we add the constraints $z_u (1-z_u) = 0$ and
$z_v (1-z_v) = 0$, and if $z_w$ was a binary variable we enforce $z_u$
and $z_v$ to be binary as well.

\end{enumerate}

Clearly, in any case we obtain that $\Gamma[\mathcal{I}_H] = H$, and
$\mathcal{I}_H$ is equivalent to $\mathcal{I}$. Note that constraints
in $\mathcal{I}_H$ involve at most \(2\) variables and the ones with
exactly \(2\) variables are linear. This invariant makes it possible to
iterate this procedure using any sequence of minor operations. 

Let $s\doteq |V(\Gamma[\mathcal{I}])|$. Corollary \ref{main-cor} implies that $\Gamma[\mathcal{I}]$ is a minor of $\mathcal{G}_{\kappa (s)}$, thus assuming the sequence of minor operations is known, we can use the procedure above to construct an instance $\mathcal{I}_{\kappa (s)}$ which is equivalent to $\mathcal{I}$ and whose intersection graph is exactly $\mathcal{G}_{\kappa(s)}$. It is not hard to see that $\mathcal{I}_{\kappa (s)}$ has the
following form (after relabeling variables):
\begin{subequations}
\begin{align}
(\mathcal{I}_{\kappa (s)}) \quad \max\, &  \sum_{i=1}^{m'} z_i \ \\
\mbox{s.t. } & a_{i,j} z_{i} + b_{i,j} z_{j} \leq d_{i,j}  && (i,j) \in E_1 \label{e1}\\
&  z_{i} = z_j && (i,j) \in E_2  \label{e2} \\
&  z_{i}(1-z_i) = 0 && i= m',\ldots, n'   \label{e3} \\
& z_i \in \{0,1\} && i \in [m']\\
& z_i \in [0,1] && i=m',\ldots, n',
\end{align}
\end{subequations}
for some appropriately defined $E_1, E_2$, and where $a_{i,j}, b_{i,j}
\in \{-1,0,1\}$, $d_{i,j} \in \{0,1\}$. 

\begin{remark} \label{remark1} Each constraint \eqref{e1} is either a
  redundant constraint (introduced with the vertex or edge deletion
  operation) or it involves at least one integer variable. This will
  be important in the next section.
 \end{remark}

\subsubsection{From approximations to exact solutions} \label{rounding}

We will now show how to construct a (truly) feasible solution from an
$\epsilon$-feasible solution to $\mathcal{I}_{\kappa (s)}$. This will
provide the link of the hardness of approximating $\mathcal{I}_{\kappa (s)}$ to the
hardness of solving $\mathcal{I}_{\kappa (s)}$ exactly.

\begin{lemma}
  Let $z \in \{0,1\}^{m'} \times [0,1]^{n'- m'}$ be an
  $\epsilon$-feasible solution to $\mathcal{I}_{\kappa (s)}$
  satisfying \eqref{eps-opt} for $\epsilon < 1/10$. Then, from $z$, we can construct $\hat{z}$ such that $\hat{z}\in \{0,1\}^{m'} \times \{0,1\}^{n'-m'}$
  is feasible and optimal for $\mathcal{I}_{\kappa (s)}$.
\end{lemma} 
\begin{proof}
  Since $z$ is an $\epsilon$-feasible solution, we have
\[ \left|z_i^2 - z_i\right| \leq 2 \epsilon  \quad  m' \le i \le n',\]
where the \(2\) arises as the \(1\)-norm of the coefficients.
Thus either $0 \leq z_i \leq 4\epsilon $ or
$|z_i - 1| \leq 4\epsilon$: $h(x) = x^2 - x$ is
decreasing in $[0,1/2)$, increasing in $[1/2,1]$, $h(0) = 0$, and
\[h(4\epsilon) + 2\epsilon = 16 \epsilon^2 - 4\epsilon + 2\epsilon 
=  16 \epsilon^2 - 2\epsilon 
= 2\epsilon (8 \epsilon - 1) 
< 0,
\]
as \(\epsilon < 1/10\).  Thus $h(4\epsilon) \leq - 2\epsilon$. From
here we conclude $z_i \leq 4\epsilon$ if $z_i\leq 1/2$. The case for
$z_i > 1/2$ is symmetric. 

Now from $z$ we construct $\hat z$ by
rounding each component to the nearest integer, and we argue the
feasibility and optimality of $\hat z$.

\begin{enumerate}[(a)]
\item Constraints \eqref{e3} are clearly satisfied as $\hat z$ is a binary vector.
\item For constraints \eqref{e2}, $z$ being $\epsilon$-feasible
  implies $ |z_i - z_j| \leq 2\epsilon $  for all $(i,j)\in E_2$ and
thus, using the above and that \(\hat{z}_i, \hat{z}_j\) are binary, we have
\begin{eqnarray*}
|\hat{z}_i - \hat{z}_j | &\leq & |\hat{z}_i - z_i | + |\hat{z}_j - z_j | + 2\epsilon \\
&\leq & 10\epsilon.
\end{eqnarray*}
The left-hand side is an integer and $\epsilon < 1/10$, from where we
conclude $\hat{z}_i = \hat{z}_j $.
\item For constraints \eqref{e1} fix $(i,j)\in E_1$ such that the
  corresponding constraint is not redundant. By Remark \ref{remark1}
  either $z_i$ or $z_j$ is integer. Without loss of generality assume
  $z_i \in \{0,1\}$, and thus $\hat{z}_i = z_i$. To make the argument
  clear, we rewrite the inequality as
\[ a_{i,j} z_{i}  - d_{i,j} \leq - b_{i,j} z_{j}. \]
The left hand side is an integer, therefore rounding $z_j$ will keep
the inequality valid where we use that \(b_{i,j} \in \{-1,0,1\}\):
\[ a_{i,j} \hat{z}_{i}  - d_{i,j} \leq - b_{i,j} \hat{z}_{j}. \]
\end{enumerate}

This proves $\hat z$ is feasible. On the other hand, $z$ satisfies \eqref{eps-opt}, and only integer variables have non-zero objective coefficient:

\[ \sum_{i=1}^{m'} \hat{z}_i  = \sum_{i=1}^{m'} z_i \geq  \sum_{i=1}^{m'} z^*_i\]
therefore $\hat{z}$ is optimal.
 \end{proof}

\subsubsection{Bringing it all together}

\begin{proof}[Main Theorem \ref{main-theo}]
Suppose we are given a sequence of graphs ${\cal
G}_k$, each having treewidth $k$. We show that, under the conditions
of Theorem \ref{main-theo}, the existence of an algorithm $\mathcal{A}$ as in Theorem \ref{main-theo} (i.e., that can approximately solve QCQP problems ${\cal I}_k$ with $\Gamma[\mathcal{I}_k] = \mathcal{G}_k$), with
running time $T(k)\cdot \textit{poly}(\|I_k \|)$ with $T(k)$ polynomial in $k$ implies that planar MAX-2SAT
belongs to $BPP$, contradicting the assumption $NP \not\subseteq BPP$.

\begin{enumerate}
\item Consider an instance of planar MAX-2SAT. We construct an instance $\mathcal{I}$ of a QCQP as in Section
  \ref{first-step}, whose intersection $\Gamma[\mathcal{I}]$ graph is planar. We denote $s$ its number of vertices.
  \item From Corollary \ref{main-cor} we know that  $\Gamma[\mathcal{I}]$ is a minor of $\mathcal{G}_{\kappa(s)}$. Moreover, $\kappa(s):=\kappa_{GM} (s/c_4)$ and, from the discussion in Section \ref{minor-operations}, $\cal I$ is equivalent to a QCQP
  problem ${\cal I}_{\kappa(s)}$ with $\Gamma[{\cal I}_{\kappa(s)}]=\mathcal{G}_{\kappa(s)}$.
  \item The minor operations transforming $\mathcal{G}_{\kappa(s)}$ into $\Gamma[\mathcal{I}]$, which are needed to construct  ${\cal I}_{\kappa(s)}$, can be obtained as follows:
  \begin{enumerate}
      \item Since $\Gamma[\mathcal{I}]$ is planar, it is a minor of the  $s/c_4 \times s/c_4$ grid. This sequence of minor operations can be found can be found in linear time using the results in \cite{tamassia}.
      \item The $s/c_4 \times s/c_4$ grid is a minor of $\mathcal{G}_{\kappa(s)}$. We can find the corresponding sequence of minor operations (with high probability) in polynomial time using the algorithm by \citet{chekuri2016polynomial} mentioned in Theorem \ref{chekuri}.
  \end{enumerate}
\item Using the point above, we can construct (with high probability) instance ${\cal I}_{\kappa(s)}$.
\item Using $\mathcal{A}$ and a fixed $\epsilon < 1/10$, find an
$\epsilon$-feasible solution satisfying \eqref{eps-opt} for
$\mathcal{I}_{\kappa(s)}$ in time $T(\kappa(s))\cdot \textit{poly}(\|I_{\kappa(s)} \|)$.
\item Given an $\epsilon$-feasible solution of $\mathcal{I}_{\kappa(s)}$, we
construct an optimal solution for $\mathcal{I}_{\kappa(s)}$ as in Section
\ref{rounding}.
\item From the optimal solution to $\mathcal{I}_{\kappa(s)}$, we can find an
optimal solution to $\mathcal{I}$ using the minor operations described
in Section \ref{minor-operations} in polynomial time.
\end{enumerate} 
Using the optimal solution, we can solve the \emph{decision problem} associated to planar MAX-2SAT directly. The only place where our algorithm can make a mistake is in the sequence of minor operations, which happens with low probability. Since clearly $\|I_{\kappa(s)}\|$ is polynomial, and by assumption we have access to $\mathcal{G}_{\kappa(s)}$ in polynomial time, if $T(\kappa(s))$ is also polynomial, we obtain that planar MAX-2SAT $\in BPP$, a contradiction.
\end{proof}


\section{Treewidth-based Extension Complexity Lower Bounds} \label{XClowerbounds}

In this section we analyze the tightness of the \emph{linear extension
  complexity} results that exploit treewidth. While we provide precise
definitions in Section~\ref{sec:backgr-extend-form}, the \emph{linear
  extension complexity} of a problem is the smallest number of
inequalities needed to represent a given problem as linear program. In
fact our lower bounds will also hold for semidefinite programs,
showing that there is little to be gained from semidefinite programs
over linear programs in terms of exploiting low treewidth. 

To this end, we consider a
set defined as
\begin{equation} \label{BPO}
S = \{ x\in \{0,1\}^n \ : \  \phi_i(x) =  0 \text{ with } i \in [m] \} 
\end{equation}
where each $\phi_i: \{0,1\}^n \rightarrow \{0,1\}$ is a boolean
function. Note that the intersection graph does not only depend on the
set $S$, but also on \emph{how} it is formulated; we denote the
intersection graph of \eqref{BPO} as $\Gamma[S_\phi]$.

\begin{remark}
  Given the generality of the $\phi_i$ functions defining the
  constraints in \eqref{BPO}, one could formulate $S$ using a
  single membership oracle of $S$. However, such a formulation would
  consist of a single constraint involving all variables, which would
  yield a very dense formulation of $S$, so that we could not exploit
  low treewidth. 
\end{remark}

  Any pure binary PO can be formulated as \eqref{BPO}. We have already
  seen in the previous section that unbounded treewidth of the
  intersection graph can yield intractability in the algorithmic
  sense. In this section, in contrast, we focus on studying how hard a
  sparse problem can be, using extension complexity as the measure of
  complexity.

\subsection{Background on Extended Formulations}
\label{sec:backgr-extend-form}
We will now briefly recall basics concepts from \emph{Extended
  Formulations} needed for our discussion. Extended formulations aim
for finding a formulation of an optimization problem in extended space
where auxiliary variables are utilized with the aim to find an overall
smaller formulations compared to formulations in the original space,
involving only the problem-inherent variables. Note that optimizing a
linear objective over an extended formulation is no harder than over
the original formulation, which makes extended formulations
appealing. For a more detailed discussion we refer the reader to
\cite{fiorini2012linear,fawzi2015positive}.

\begin{Def}[Linear Extended Formulation]
Given a polytope $P\subseteq \mathbb{R}^n$, a \emph{linear} extended formulation
of $P$ is a linear system 
\begin{equation}\label{linearEF}
E x + F y = g,\quad y\geq 0
\end{equation}
with the property that $x\in P$ if and only if there exists $y$ such that $(x,y)$ satisfies \eqref{linearEF}. The \emph{size} of the linear
extension is given by the number of inequalities in \eqref{linearEF},
and the \emph{linear extension complexity} of $P$ is the minimum size
of a linear extended formulation of $P$, which we denote by $xc(P)$.
\end{Def}

\begin{remark}
In the previous definition, system \eqref{linearEF} can be made more general. 
We can also consider 
\[E x + F y = g^{=}, \quad E^{\leq} x + F^{\leq} y \leq g^{\leq}\]
and define the \emph{size} the same way as before. However, this more
general definition does not affect the extension complexity of a
polytope; see e.g., \cite{Yannakakis}.
\end{remark}

In Yannakakis' ground-breaking paper \cite{Yannakakis}, it is proved that the linear extension complexity of a polytope
is strongly related to the concepts of \emph{slack matrix} and \emph{non-negative rank}:

\begin{Def}[Slack Matrix]
Let $P$ be a polytope that can be formulated as
\[P = \{ x\in \mathbb{R}^n \ : \ a_i^T x \leq b_i, \ i \in [m]\}.\]
Consider a set of points $V = \{x_j : j\in J\}$ such that $P = \conv(V)$. Then, the \emph{slack matrix} $S$ of $P$ associated to $Ax \leq b$ and $V$ is given by
\[S_{ij} = b_i - a_i^T x_j.
\]
\end{Def}

\begin{Def}[Non-negative Factorization]
Given a non-negative matrix $M$, a rank-$r$ non-negative factorization of $M$ is given by two non-negatives matrix $T$ (of $r$ columns) and $U$ (of $r$ rows) such that 
\[ M = TU. \]
The \emph{non-negative rank} of $M$, denoted as $rk_+(M)$, is the minimum rank of a non-negative factorization of $M$.
\end{Def}

\begin{theorem}[\citet{Yannakakis}]
Let $P = \{ x\in \mathbb{R}^n \ : \ Ax\leq b \} = \conv(V)$ be a polytope with
$\dim(P) \geq 1$ and let $S$ be the slack matrix of $P$ associated to $Ax \leq b$ and $V$. Then
\[xc(P) = rk_+(S).\]
\end{theorem}

In the linear case the $y$ variables in the extended formulation are
required to be in the cone given by the non-negative orthant, i.e.,
\(y \geq 0\). This was generalized to other cones, allowing for more
expressiveness in the extended space.  Of particular interest to this
work is the generalization to \emph{semidefinite extended
  formulations}; see \cite{fiorini2012linear,gouveia2013lifts} for
details on the following concepts and results.

\begin{Def}[Semidefinite Extended Formulations]
Given a convex set $K\subseteq \mathbb{R}^n$, a \emph{semidefinite} extended formulation
of $K$ is a system
\begin{equation}\label{semidefiniteEF}
a_i^T x + \langle U_i, Y \rangle = b_i, \, i\in I, \quad Y\in \mathbb{S}^r_+
\end{equation}
where $I$ is an index set, $a_i\in \mathbb{R}^n$,
$U_i \in \mathbb{S}^r_+$, with the property that $x\in K$ if and only
if there exists $Y$ such that $(x,Y)$ satisfies
\eqref{semidefiniteEF}. The \emph{size} of the semidefinite extension
is given by the size $r$ of matrices $U_i$ in \eqref{semidefiniteEF},
and the \emph{semidefinite extension complexity} of $K$ is the minimum
size of a semidefinite extended formulation of $K$. It is denoted
$xc_{SDP}(K)$.
\end{Def}

\begin{Def}[Semidefinite Factorization]
Given a non-negative $n\times m$ matrix $M$, a rank-$r$ semidefinite factorization of $M$ is given by a set of pairs $(U_i, V^j)_{(i,j)\in [n]\times [m]} \subseteq \mathbb{S}^r_+ \times \mathbb{S}^r_+ $ such that
\[ M_{i,j} = \langle U_i, V^i \rangle \quad \forall i \in [n], j
\in [m]. \]
The \emph{semidefinite rank} of $M$, denoted as $rk_{PSD}(M)$, is the minimum rank of a semidefinite factorization of $M$.
\end{Def}

\begin{theorem}[Yannakakis' Factorization Theorem for SDPs, \cite{gouveia2013lifts}] \label{yannakakisPSD}
Let $P = \{ x\in \mathbb{R}^n \ : \ Ax\leq b \} = \conv(V)$ be a polytope with
$\dim(P) \geq 1$ and let $S$ be the slack matrix of $P$ associated to $Ax \leq b$ and $V$. Then
\[xc_{SDP}(P) = rk_{PSD}(S).\]
\end{theorem}

Note that every linear extended formulation is a semidefinite extended
formulation using diagonal matrices so that $xc_{SDP}(P) \leq
xc(P)$. 

\subsection{Low treewidth implies small extension complexity} \label{sec:lowtreewidth}
We will now state the known upper bound on the linear extension
complexity of low-treewidth problems, which we prove to be nearly
optimal. The following strong result is well known; see e.g.,
\cite{treewidth,kolman2015extension,Laurentima}:

\begin{theorem} \label{theo:binary}
Let $S\subseteq \{0,1\}^n$ be a set that exhibits a \emph{formulation} as
\begin{equation} \label{formulation}
S = \{ x\in \{0,1\}^n \ : \  \phi_i(x) =  0 \text{ with } i \in [m] \}.
	\end{equation}
 If $\Gamma[S_\phi]$ has treewidth $\omega$, then $\conv(P)$ has linear extension complexity 
\begin{equation} \label{xcbinary}
O(n2^\omega).
\end{equation}
\end{theorem}

We will construct sets $S$ that
\begin{inparaenum}[(a)]
\item can be formulated
using sparse constraints (given by some treewidth $\omega$) and which
\item exhibit high extension complexity essentially of \eqref{xcbinary}.
\end{inparaenum}
By
building on recent lower bounds on semidefinite extension complexity
\cite{briet2013existence}, we show the existence of such 0/1 sets,
whose semidefinite extension complexity (nearly) meets the bound
\eqref{xcbinary} (see Main Theorem \ref{thr:counting}).  In fact, for those
hard instances, we show a stronger result. The extension complexity
does not take into account techniques that are routinely
adopted to solve integer programs, such as e.g., reformulations or parallelization of separable sets. These techniques can be used to modify the original instance to an equivalent integer programming problem, which may be computationally more attractive. We show that the hard instances we construct cannot be reformulated to have lower extension complexity or being separable.

The careful reader might have noticed an important fact: the extension
complexity bound in \eqref{xcbinary} does not depend on a particular
formulation of the set $S$, as opposed to the treewidth. To overcome
this disparity and for simplicity in the upcoming discussion we focus
on the ``best possible'' treewidth of a formulation, which we refer to
as the \emph{treewidth (or treewidth complexity) of \(S\)}. This
definition prevents the results from
depending on a particular formulation, or the type of constraints
(e.g., linear, boolean, or polynomial).

\begin{Def} \label{treewidthofset} Given $S \subseteq \{0,1\}^n$, we
  denote as $tw(S)$ the smallest treewidth of the intersection graph
  of any formulation of $S$ as in \eqref{formulation}.
\end{Def}

\begin{remark} It came to our attention that, independently
of this work, in \citet{aboulker2019extension} it was recently proven that for
any minor-closed family of graphs there exists a constant $c$ such
that the \emph{correlation polytope} of each graph of $n$ vertices in
the minor-closed family has linear extension complexity at least
\begin{equation}\label{corrupper} 2^{c(\omega + \log n)}
\end{equation} where $\omega$ is the treewidth of the graph. This provides
families of polytopes where \eqref{xcbinary} is almost tight. While
this result is in the same spirit as the result we prove in this section, we highlight a few key differences:
\begin{enumerate}
\item The results in \cite{aboulker2019extension} study the important question of the linear extension complexity of the correlation polytope for various graphs providing (almost) optimal bounds, while we give ourselves more freedom with the polytope family.
\item The constant $c$ in \eqref{corrupper} is at most $1/2$, and the
correlation polytope of a graph with treewidth $\omega$ has ambient dimension
$N \in O( \omega n)$---the number of edges of the graph. If additionally $N \in \Theta(\omega n)$, the lower bound in \eqref{corrupper} satisfies
\[ 2^{c(\omega + \log n)} \in O\left( \sqrt{ \frac{N}{\omega} } 2^{\omega/2} \right).\]
The polytopes we construct here have a lower bound with a leading term
$N/\omega$ as compared to \(\sqrt{ \frac{N}{\omega}}\). This is due to the fact
that we rely on the stronger existential counting arguments in
\cite{rothvoss2013some,briet2013existence,briet2015existence} along with a polytope composition procedure.
\item Our employed technique is drastically different: rather than
  \emph{reducing} to a face of the correlation polytope we provide a
  general technique to \emph{construct} high-extension-complexity polytopes
  from any seed polytope (under appropriate assumptions). 
\item Our results apply to both the \emph{semidefinite}
and the \emph{linear} case. Moreover, we also specialize our construction
to Stable Set polytopes where the gluing operation that we use has a
natural representation in terms of graph-theoretic operations.
\end{enumerate}
\end{remark}

\subsection{Binary optimization problems with high extension complexity} \label{sec:highxc}

In this section we analyze how high semidefinite extension complexity can be used to derive characteristics of the formulation of sets and their
treewidth. Consider a family of sets
$\{S_n\}_{n\in \mathbb{N}}$ with $S_n \subseteq \{0,1\}^n$ such that
\begin{equation}
    \label{SnDef}
    xc_{SDP}(S_n) \in \Omega\left( 2^{f_n}\right)
\end{equation}
for some $f_n$. For technical
reasons we further assume that $f_n$ satisfies
\begin{equation}\label{technical-bound}
\liminf_{n\rightarrow \infty }\frac{\log n}{f_n} < 1.
\end{equation}

\begin{remark}
Every family of sets $S_n$ such that $xc_{SDP}(S_n) \in \Omega(n^k)$ for some $k> 1$ satisfies \eqref{technical-bound}. In such case, $f_n \geq k \log n$ asymptotically and \eqref{technical-bound} can be easily verified.
\end{remark}

Assuming \eqref{technical-bound} only excludes sets with linear or
sub-linear semidefinite extension complexity (w.r.t.~\(n\)), which are
of little interest here. Moreover, by \cite{briet2013existence}, we
know there exist 0/1 sets whose semidefinite extension complexity
satisfies \eqref{technical-bound}.

\begin{lemma}\label{lemma:big-treewidth-general}
  Any formulation of $S_n$ has intersection graph with treewidth
  $\Omega(f_n)$ and at most $n-1$. In particular, $tw(S_n)$ is $\Omega(f_n)$
  and $O(n)$.
\end{lemma}
\begin{proof}
The upper bound is immediate, since $S_n$ has $n$ variables. For the lower bound,
we know from Theorem \ref{theo:binary} there exists $c_1$ such that 
\begin{equation}
xc_{SDP}(S_n) \leq c_1 n 2^{\omega_n }, 
\end{equation}
where $\omega_n$ is the treewidth obtained from a formulation \eqref{formulation}. And since 
\begin{equation}
xc_{SDP}(S_n) \geq c_2 2^{f_n}
\end{equation}
for some $c_2$, we obtain
\[f_n \leq \log (c_1/c_2) + \omega_n + \log n. \]

If $\omega_n \in o(f_n)$ this implies
\[ 1 \leq \liminf_{n\rightarrow \infty} \frac{\log n}{f_n},\]
a contradiction with \eqref{technical-bound}. We conclude
$\omega_n \in \Omega(f_n)$.
 \end{proof}

\subsection{Composition of Polytopes} \label{sec:composition}

The techniques in this section allow us to manipulate the sets $S_n$
in a convenient way. Here we drop the index $n$ for ease of notation
as all definitions and results apply for any 0/1 set. We use the
notation \(\alpha S\) with \(\alpha \in \R_+\) to denote the set \(\{x \mid
x = \alpha \cdot y \text{ with } y \in S\}\); in particular \(0 \in
0S\) for all \(S\). 

\begin{Def}\label{def:plusoperator}
For $S \subseteq \{0,1\}^n$, we define $S^{+} \subseteq \{0,1\}^{n+1}$ as
\[ S^+ = \{(x,x_{n+1}) \in \{0,1\}^{n+1} \ | \ x\in (1-x_{n+1})S \}. \]
\end{Def}

In particular, \((x,0) \in S^+\) for all \(x \in S\) and \(e_{n+1} \in
S^+\). We obtain the following lemma:

\begin{lemma}\label{lemma:conv-Pplus}
\begin{equation}\label{conv-Pplus}
\conv(S^+) = \{(x,x_{n+1}) \in [0,1]^{n+1} \ | \ x\in (1-x_{n+1})\conv(S) \}.
\end{equation}
\end{lemma}
\begin{proof}
Inclusion $\subseteq $ is direct, as the right-hand set is convex, and
the inclusion can be directly verified for the extreme points.

Now consider $(x,x_{n+1})$ an extreme point of the right-hand set in \eqref{conv-Pplus}. We first claim $x_{n+1} \in \{0,1\}$. Otherwise,
we can write
\[(x,x_{n+1}) = (1-x_{n+1}) (x/(1-x_{n+1}) , 0) + x_{n+1}e_{n+1},\]
where $e_{n+1}$ is the $(n+1)$-th canonical vector. By assumption
$x/(1-x_{n+1})\in S$ thus $(x/(1-x_{n+1}) , 0)\in S^+$ and
$e_{n+1} \in S^+$. This contradicts $(x,x_{n+1})$ being an extreme
point.

As such $x_{n+1} \in \{0,1\}$ and we can easily verify that
$(x,x_{n+1})\in S^+$ which proves the remaining inclusion.
 \end{proof}

\begin{Def} \label{def:pyramid} A polytope $Q \subseteq \mathbb{R}^n$
  is called a \emph{pyramid with base $B\subseteq \mathbb{R}^n$ and
    apex $v\in \mathbb{R}^n$} if
\[ Q = conv(B\cup \{v\}) \]
and $v$ is not contained in the affine hull of $B$.
\end{Def}

In \citet{tiwary2017extension} the extension complexity of the
Cartesian product of polytopes is analyzed and it is shown:

\begin{theorem} \label{theo:pyramids}
Let $Q_1, Q_2$ be non-empty polytopes such that one of the two polytopes is a pyramid. Then
\[xc(Q_1\times Q_2) = xc(Q_1) + xc(Q_2) \]
\end{theorem}

This result provides us with a tool to combine polytopes in a way that
their extension complexity is added up. Unfortunately, the result is
limited to \emph{linear} extended formulations. We generalize this
result to the SDP case here:

\begin{theorem} \label{theo:pyramidsSDP}
Let $Q_1, Q_2$ be non-empty polytopes such that one of them is a pyramid. Then
\[xc_{SDP}(Q_1\times Q_2) \geq xc_{SDP}(Q_1) + xc_{SDP}(Q_2) - 1 \]
\end{theorem}

\begin{proof}
  This result follows directly from combining the analysis by
  \citet{tiwary2017extension} with a result from
  \citet{fawzi2015positive}. We assume w.l.o.g.~that $Q_2$ is a
  pyramid and thus we may assume the slack matrix $T$ of $Q_2$ has
  the form
\[T = \left[\begin{array}{c|c}
 T' & 0 \\
\hline
0 & 1
\end{array} \right] \]
with $T'$ a slack matrix of the base $Q'_2$ of $Q_2$. This implies
\begin{equation}\label{pyramid-base-PSD}
xc_{SDP}(Q_2) = xc_{SDP}(Q'_2) + 1,
\end{equation}
(see e.g., \cite[Theorem 2.10]{fawzi2015positive}). On the other hand,
it also implies that there is a slack matrix $A$ of $Q_1 \times Q_2$
of the following form (see \cite{tiwary2017extension}):
\[A = \left[\begin{array}{c|c|c|c}
S & \cdots & S & S \\
\hline
t'_1 \cdots t'_1 & \cdots & t'_k \cdots t'_k & 0 \cdots 0 \\
\hline
0 \cdots 0 & 0 \cdots 0 & 0 \cdots 0 & 1 \cdots 1
            \end{array} \right], \]

 where each $t'_i$ corresponds to a column of $T'$ and $S$ is
 a slack matrix of $Q_1$. Further, the following matrix is a
 sub-matrix
 of $A$:

\[A' = \left[\begin{array}{c|c}
 S & S \\
\hline
T' & 0
\end{array} \right]. \]
Since this is a block-triangular matrix by \cite[Theorem 2.10]{fawzi2015positive}, we know that
\[rank_{PSD}(A') \geq rank_{PSD}(S)  + rank_{PSD}(T'). \]
Using the factorization theorem for semidefinite extended formulations
(Theorem~\ref{yannakakisPSD}) and \eqref{pyramid-base-PSD} we obtain
\[xc_{SDP}(Q_1\times Q_2) \geq xc_{SDP}(Q_1) + xc_{SDP}(Q_2) -
1.\] \end{proof}

The previous result will allow us to combine polytopes and
obtain a lower bound for the resulting extension complexity. To this end
we prove the following:

\begin{lemma} \label{general-pyramid}
Let $S$ and $S^+$ be as before. Then
\begin{enumerate}
\item[(a)] $\conv(S^+)$ is a pyramid with base $\conv(S)\times \{0\}$ and apex $e_{n+1}$.
\item[(b)] $tw(S) \leq tw(S^+) \leq tw(S) + 1$.
\end{enumerate}
\end{lemma}
\begin{proof}
\begin{inparaenum}
Property \item[(a)] follows directly from the proof of Lemma \ref{lemma:conv-Pplus}.
For property \item[(b)] consider a formulation
\begin{equation}\label{Pformulation}
 S = \{x\in\{0,1\}^n \ | \ \phi_i(x) = 0, \ i \in [m]\}.
\end{equation}
Then a valid formulation for $S^+$ is given by
 \[ S^+ = \{ (x, x_{n+1})\in \{0,1\}^{n+1} \ | \ (1-x_{n+1})\phi_i(x)
 = 0,\ i \in [m] \text{ and } \ x_j \leq 1- x_{n+1} ,\ j \in [n]\}.\]
(note that inequalities can be interpreted as boolean functions
as well). This formulation of $S^+$ has an intersection graph
formed by adding a new vertex to the intersection graph
of \eqref{Pformulation} connected to every other vertex.
This increases the treewidth by at most \(1\) and hence
\[tw(S^+) \leq tw(S) + 1.\]

For the remaining inequality, take a formulation of
$S^+$ whose intersection graph has minimal treewidth:
\begin{equation}
\label{eq:formulation1}
S^+ = \{(x,x_{n+1}) \in \{0,1\}^{n+1}\ | \ \varphi_i(x,x_{n+1}) = 0,\
i \in [m]\}.
\end{equation}
Since $S = \{x \in \{0,1\}^{n} \ | \ (x,0) \in S^+\}$, we obtain
\begin{equation}
    \label{eq:formulation2}
    S =  \{x \in \{0,1\}^{n}\ |\ \varphi_i(x, 0) = 0,\ i \in [m]\}.
\end{equation}
The treewidth associated with formulation \eqref{eq:formulation2} is at most the treewidth of formulation \eqref{eq:formulation1}, as the intersection graph of the former is obtained by removing
a vertex from the intersection graph of the latter. By assumption,
the treewidth of formulation \eqref{eq:formulation1}
is $tw(S^+)$, thus 
\[tw(S) \leq tw(S^+).\]
\end{inparaenum}
 \end{proof}

In what follows, we will need a short technical lemma.

\begin{lemma}\label{technical}
Let $S \subseteq \{0,1\}^n$.  Then
\[tw(S\times S) = tw(S). \]
\end{lemma}
\begin{proof}
    Inequality $\leq$ follows directly, since any formulation of $S$
    can be used to formulate $S\times S$. Moreover, the intersection
    graph of such formulation consists of \(2\) identical copies of the
    intersection graph of the formulation of $S$. From here the
    inequality follows.
    
    For the other inequality, take any formulation for $S\times S$:
    \begin{equation}\label{formulationPxP}
        S\times S = \{(x,y) \in \{0,1\}^{2n}\ : \ \varphi_i(x,y) = 0,\
        i= \in [m].\}
    \end{equation}
    Let $\hat{y} \in S$ be arbitrary. By definition we must have that
    $x\in S$ if and only if $(x,\hat{y}) \in S\times S$, thus
    $S = \{ x\in \{0,1\}^n\ : \ \varphi_i(x,\hat{y}) = 0,\ i=1,\ldots,
    m \}$ is a valid formulation for $S$. The intersection graph of
    such formulation is a sub-graph of the intersection graph of
    formulation \eqref{formulationPxP}, thus its treewidth is
    at most as large. This proves $tw(S)\leq tw(S\times S)$.  \end{proof}

The results above shows the key fact that taking Cartesian product of certain polytopes adds up their extension complexity, but roughly maintains their treewidth. We summarize this in the following Lemma.
\begin{lemma}\label{general-products}
Let $S \subseteq \{0,1\}^n$ and define
\[S^{\times k} = S^+ \times \cdots \times S^+\]
where the Cartesian product is taken $k$ times. Then
\[xc_{SDP}(S^{\times k}) \geq k\cdot xc_{SDP}(S) \]
and
\[ tw(S) \leq tw(S^{\times k}) \leq tw(S) + 1.\]
\end{lemma}
\begin{proof}
Since $\conv(S^+)$ is a pyramid (part (a) of Lemma \ref{general-pyramid}) and $\conv(S^+\times S^+) = \conv(S^+) \times \conv(S^+)$ we obtain
\begin{align*}
xc_{SDP}(S^{\times k})  & \geq  xc_{SDP}(S^{\times (k-1)}) + xc_{SDP}(S^+) - 1 && \text{(by Theorem \ref{theo:pyramidsSDP})}\\
& =  xc_{SDP}(S^{\times (k-1)}) + xc_{SDP}(S)  && \text{(by \eqref{pyramid-base-PSD})}
\end{align*}
Applying this inductively we obtain $xc_{SDP}(S^{\times k}) \geq k\cdot xc_{SDP}(S)$.
On the other hand, applying Lemma \ref{technical} iteratively we have
\[tw(S^{\times k}) = tw(S^+)\]
and thus the treewidth claim follows from part (b) of Lemma \ref{general-pyramid}.
\end{proof}

\subsection{Composing polytopes of high semidefinite extension complexity}
We now use the results in Sections \ref{sec:highxc} and
\ref{sec:composition} and a family $\{S_n\}_{n\in \mathbb{N}}$ of
(assumed) high (semidefinite) extension complexity, to construct a family of
polytopes having a (semidefinite) extension complexity lower bounded by treewidth.

\begin{theorem}
  \label{theo:general} Let $\{S_n\}_{n\in \mathbb{N}}$ be a family of sets satisfying \eqref{SnDef}, i.e.
    \[xc_{SDP}(S_n) \in \Omega\left( 2^{f_n}\right),\]
    and technical condition \eqref{technical-bound}. Consider a sequence $\{\omega_n\}_{n \in \mathbb{N}}$ with
  $\omega_n\leq n-1$ for all $n \in \mathbb{N}$. Then there exists a family
  of sets $\{S'_n\}_{n \in \mathbb{N}}$, $S'_n \subseteq \{0,1\}^n$, such that:
	\[tw(S'_n) \leq \omega_n + 1 \quad \text{and}\quad xc_{SDP}(S'_n)\in \Omega \left( \frac{n}{\omega_n + 1}  2^{f_{\omega_n}}\right). \]
	Moreover, $\conv(S_n')$ is a pyramid and $tw(S'_n) \in \Omega(f_{\omega_n})$.
\end{theorem}

\begin{proof}
Fix $n\in \mathbb{N}$ and consider set $S_{\omega_n}$.  This set has $\omega_n$ variables and from
Lemma \ref{lemma:big-treewidth-general} $tw(S_{\omega_n})$ is $\Omega(f_{\omega_n})$ and at most $\omega_n - 1$.
Now let $k\in \mathbb{N}$ and consider $S_{\omega_n}^{\times k}$. By Lemma \ref{general-products}
\[tw(S_{\omega_n}) \leq tw(S_{\omega_n}^{\times k}) \leq tw(S_{\omega_n}) + 1\]
which implies $tw(S_{\omega_n}^{\times k})$ is at most $\omega_n$. 
Additionally 
\[ xc_{SDP}(S_{\omega_n}^{\times k}) \geq k\cdot xc(S_{\omega_n}) \in \Omega\left( k \cdot  2^{f_{\omega_n}} \right).\]
As a last step, we define 
\begin{equation}\label{lastpyramid}
S'_n = \left(S_{\omega_n}^{\times k}\right)^+,
\end{equation}
which inherits the extension complexity bounds from
$S_{\omega_n}^{\times k}$ and increases the treewidth by at most 1. The last requirement we need is 
 $S_n'$ to have at most $n$ variables, hence, we require
 \[ k \cdot (\omega_n +1) \leq n - 1.\]
Choosing $k = \lfloor \frac{n - 1}{\omega_n + 1} \rfloor$ concludes
the result.
 \end{proof}

The reader might notice that the last step taken in \eqref{lastpyramid} 
is not necessary to obtain the extension complexity result. However,
this will prove useful next, when we further analyze how hard these
instances are.

We are now ready to apply the techniques we developed to some known
hard polytopes, thus showing that Theorem \ref{theo:binary} is
essentially tight.

\begin{maintheorem}\label{thr:counting}
    For every $\{\omega_n\}_{n \in \mathbb{N}}$ satisfying $\omega_n\leq n - 1$ for all $n \in \mathbb{N}$, there exists a family of sets $\{S'_n\}_{n \in \mathbb{N}}$ each with at most $n$ variables and such that:
	\[ tw(S'_n) \leq \omega_n + 1 \quad \text{and} \quad xc_{SDP}(S'_n)\in \Omega \left( \frac{n}{\omega_n + 1}  2^{\frac{\omega_n}{4}(1-o(1))}\right) \]
Moreover, $tw(S'_n)$ is $\Omega(\frac{\omega_n}{4}(1-o(1)))$ and $\conv(S_n')$ is 
a pyramid.
\end{maintheorem}
\begin{proof}
  In \cite{briet2013existence} the the existence of $n$-dimensional
  0/1 polytopes with semidefinite extension complexity lower bounded
  by
\[2^{\frac{n}{4}(1-o(1))},\]
is shown. We simply use the vertices of these polytopes as
$\{S_n\}_{n\in\mathbb{N}}$ in Theorem \ref{theo:general} and the result is obtained. 
 \end{proof}

Note that Theorem \ref{thr:counting} provides a nice additional insight:
as $tw(S'_n)\in O(\omega_n)$ the instances we construct can be
formulated \emph{sparsely}, but there is no valid formulation that is
considerably \emph{sparser} than that as $tw(S'_n)$ is
$\Omega(\frac{\omega_n}{4}(1-o(1)))$.

\subsection{Reformulations}\label{sec:reformulations}
When solving optimization problems in general and integer
programming problems in particular, reformulation techniques are often
employed to modify the original instance, in order to obtain a more
well-behaved one. For instance, the affine map $(x_1,x_2-1)$ can be
applied to the set $\{(0,0),(0,1),(1,1),(2,1)\}\subseteq \R^2$ to
obtain the set $\{(0,0),(0,1),(1,0),(1,1)\}$. The convex hull of the
latter set can then be ``decomposed'' as the Cartesian product of the
line segments $[0,1]$ and $[0,1]$; the same is not
true for the former.

In this section we show that the hard instances from Theorem \ref{thr:counting} are
robust with respect to common reformulations techniques. A very general notion of reformulation was introduced by
\citet{BraunPokutta}, where the authors deal with any nonnegative problem and allow to customize which objective functions (called \emph{evaluation}) the original problem and the reformulation have to agree on (see \cite{BraunPokutta} for details). Here, we restrict the definition to reformulations that agree with the original problem for \emph{any} nonnegative objective function.

\begin{Def}\label{def:reformulation} 
  Let ${S}\subseteq \{0,1\}^n$. We define a \emph{reformulation}
  of ${S}$ as a triple $({ S}',f,d)$, where:

\begin{enumerate}
\item[(a)]\label{eq:setS} ${S}\subseteq \mathbb{R}^m$ is an arbitrary set;
\item[(b)] $f:{S}\rightarrow {S}'$ is a bijection;
\item[(c)]\label{eq:obj} $d$ is a collection of affine functions
  $\{d_c: \R^m \rightarrow \R: c \in {\cal C}\}$ with the property that $d_c(f(x))=c(x)$ for all
  $x \in {S}$ and $c \in {\cal C}$, where
  ${\cal C}$ contains all affine functions $c: \R^n \rightarrow \R$ with \(c(x) \geq 0\) for all
  $x \in { S}$.
\end{enumerate}
\end{Def}

The motivation for this definition is that, for each \emph{affine} function $c$ that is nonnegative over ${ S}$, one could find the optimal
solution to the instance
\begin{equation}\label{eq:a-BP} \max \{c(x) :  x \in { S}\} \end{equation}
by finding an optimal solution $y^*$ to the
instance \begin{equation}\label{eq:BP-reformulation}\max \{d_c(y) : y
  \in { S}'\}\end{equation} and then outputting $f^{-1}(y^*)$. Typically one
  is  interested in the case when \eqref{eq:BP-reformulation} is an
integer programming problem, thus $f$ can be viewed as in integer programming re-encoding
of an optimization problem. However, in the following analysis we do not need to restrict
$S'$ to be an integer set. Furthermore, we do not need to assume $S\subseteq \{0,1\}^n$ either, but we
will phrase everything with this assumption as the sets we construct in this work are all 0/1 valued.

For a set ${S}\subseteq \{0,1\}^n$, let $xc^{ref}({S})$ (resp. $xc_{SDP}^{ref}({S})$) be the minimum linear (resp. semidefinite) extension complexity of a reformulation for ${S}$. Clearly $xc^{ref}({S})\leq xc({S})$ and $xc_{SDP}^{ref}({S})\leq xc_{SDP}({S})$, as the set itself can be viewed as a trivial reformulation.

The following results appeared in \cite{BraunPokutta}, and show that the extension complexity of a set cannot be significantly reduced using a reformulation.

\begin{theorem}
$xc^{ref}({S})\geq  xc({S}) - 1 \quad \hbox{ and } \quad
xc^{ref}_{SDP}({S}) \geq  xc_{SDP}({S}) - 1. $
\end{theorem}

Thus, reformulating problems as in Definition \ref{def:reformulation} can decrease their extension
complexity by at most $1$. This shows the sets we construct are robust, in terms of their extension complexity, for reformulations. However, extension complexity is not the only parameter to measure how hard a problem is, and \emph{separability} can be also used to achieve tractability. We define and analyze this next.

\subsubsection{Reformulations and separability}

Given that the extension complexity cannot be reduced significantly via reformulations, one could aim at reformulating an optimization problem in a way that the resulting set is a Cartesian product of lower-dimensional sets (see the example at the beginning of Section \ref{sec:reformulations}).

More formally, for a $d$-dimensional set ${S}\subseteq \mathbb{R}^n$, we say it is
\emph{decomposable} if we can write
${S} = {S}_1 \times {S}_2$, with
${S}_i \subseteq \mathbb{R}^{n_i}$ (with \(i \in \{1,2\}\)) of
dimension $d_i$ so that $d_1 + d_2 = d$, and $d>d_1 \geq d_2$. We say
that a reformulation $({S}',f,d)$ of ${S}$ is decomposable
if so is ${S}'$. In this section, we prove that the hard instances from Theorem
\ref{thr:counting} cannot be reformulated to be decomposable. 

\begin{lemma}\label{lem:pyramids-stay-pyramids}
Let  $({S}',f,d)$ be a reformulation of ${S}\subseteq \{0,1\}^n$. If $\text{conv}({S})$ is a pyramid, then 
 $\text{conv}({S}')$ is also a pyramid.
\end{lemma}
\begin{proof}
Since $\text{conv}({S})$ is a pyramid, there exist ${ B}\subseteq {S}$ and a point $v\in {S}\setminus \text{aff}({ B})$ such that 
 \[\text{conv}({S}) = \text{conv}({ B}\cup \{v\}).\] Hence, there exists an affine function $\hat{c}: \R^n \rightarrow \R$ such that \[ 0\leq M = \hat{c}(x)\neq \hat{c}(v)\geq 0 \quad \forall x\in {S}\cap { B}.\]
 
 We claim that $ \text{conv}({S}')$ is a pyramid with base $ \text{conv}(\{f(x) \, :\, x\in {S} \cap { B}\})$ and apex $f(v)$. This follows from the fact that $d_{\hat{c}}( f(x)) = \hat{c}(x) $, thus
 $M = d_{\hat{c}}(f(x)) \neq d_{\hat{c}}( f(v) )\quad \forall x\in {S}\cap { B} $. \end{proof}

Next we show that pyramids are not decomposable. 

\begin{lemma}\label{lem:roma}
Let $\conv({S})\subseteq \R^n$ be a pyramid. Then ${S}$ is not decomposable.
\end{lemma}
\begin{proof} Without loss of generality let ${S}$ have dimension
  $d \geq 3$; the statement is trivial otherwise. Assume for
  contradiction that ${S}= {S}_1 \times {S}_2 $ with
  ${S}_i \subseteq \mathbb{R}^{n_i}$ of dimension $d_i$ for
  $i \in \{1,2\}$ and without loss of generality let \(d_1 \geq
  d_2\). Let $(\bar y_1,\bar y_2)$ be the apex of $\conv({
    S})$. Since $d_1\geq d_2 = d - d_1 \geq 1$, we deduce that there
  exist $\tilde y_1 \in \mathbb{R}^{n_1}$,
  $\tilde y_2 \in \mathbb{R}^{n_2}$ with $\tilde y_i \neq \bar y_i$
  for $i \in \{1,2\}$ such that $(\tilde y_1, \tilde y_2)$ is an extreme point of
  $\conv({S})$. Moreover,
  $(\bar y_1,\tilde y_2), (\tilde y_1,\bar y_2),(\tilde y_1,\tilde
  y_2)$ are all extreme points, and since they are different from the
  apex, they must lie on the base of $\conv({S})$. Nonetheless,
  the affine space generated by those three latter points contains
  $(\bar y_1,\bar y_2)$, a contradiction. \end{proof}

Note that the hypothesis of the previous lemma cannot be relaxed to the weaker assumption that $\conv({S})$ \emph{(only) contains} a pyramid, as the example from the beginning of Section
\ref{sec:reformulations} shows.

\begin{lemma}
Fix $n \in \mathbb{N}$, let ${S}:=P'_n \subseteq \{0,1\}^n$ as in Theorem \ref{thr:counting}. Then ${S}$ does not admit a reformulation that is decomposable.
\end{lemma}

\begin{proof}
  Let $({S}',f,d)$ be a reformulation of ${S}$. By Lemma 
  \ref{lem:pyramids-stay-pyramids}, $\conv({S}')$ is a
  pyramid and  by Lemma \ref{lem:roma}, ${S}'$ is not decomposable.
 \end{proof}

\section{Related Results} \label{extra-results}
We will now present several related results. 

\subsection{Specialization to Stable Set Polytopes}

We proved the existence of certain 0/1 polytopes with high exponential
semidefinite extension complexity, parametrized using the treewidth of
a formulation of the set itself in Theorem \ref{thr:counting}. For this
purpose, we used a 0/1 set that does not necessarily correspond to a
polytope of a combinatorial problem, such as, e.g., the \emph{stable
  set problem} or the \emph{matching problem}. And even if we had used a
family of combinatorial polytopes as a starting point, there is no
guarantee that the resulting polytopes in Theorem \ref{thr:counting}
correspond to a combinatorial problem as well. In this section we show
that the argument in Theorem \ref{theo:general} is compatible with the
stable set polytope, and one can state a similar parametrized lower
bound on the semidefinite extension complexity of a family of stable
set polytopes. Due to the restriction of the class of polytopes
considered here, the lower bound is worse than that of Theorem
\ref{thr:counting}, but nonetheless it is exponential in the treewidth
parameter.

\begin{Def}
Given a graph $G = (V,E)$ on $n$ nodes, where $V= [n]$. We define
\begin{equation}
    \text{STAB}(G) = \{x\in \{0,1\}^n \ | \ x_i + x_j \leq 1 \ \forall \{i,j\} \in E(G)\}. \label{eq:stab}
\end{equation}
\end{Def}

We first note there is a correspondence between the treewidth of the
graph $G$, and the treewidth of the set $\text{STAB}(G)$ as defined in
Definition \ref{treewidthofset}. While expected this is a non-trivial
fact, since one could conceive the existence of a
boolean-formula-based formulation of $\text{STAB}(G)$ that can be
sparser that $G$ itself. We prove that this is not the case.

\begin{lemma} Given a graph $G = (V,E)$. Then:
\[tw(G) = tw(\text{STAB}(G))\]
\end{lemma}

\begin{proof}
  Note that $tw(G) \geq tw(\text{STAB}(G))$, since the formulation
  given in \eqref{eq:stab} has as intersection graph $G$ itself. For
  the $\leq$ inequality, we prove that the intersection graph of
  \emph{any} formulation of $\text{STAB}(G)$ has $G$ as subgraph. For
  contradiction, suppose there exist $\phi_i$ with  $i \in [m]$ such
  that
\[ \text{STAB}(G) = \{x\in \{0,1\}^n \ |\ \phi_i(x) = 1, \ i \in [m]\}.
\]
and that $G$ is not a subgraph of $\Gamma[\text{STAB}(G)_\phi] = (V,E')$. This implies that there must exist $\{k,l \} \in E $ such that $\{k,l \} \not\in E'$. Defining
\[I_j \doteq \{i \, : \, x_j\in  \text{supp}(\phi_i) \} \]
with $\text{supp}(\phi_i)$ being the set of variables that appear explicitly
in $\phi$, we obtain $I_k \cap I_l = \emptyset$. On the other hand,
both $e_k$ and $e_l$---the $k$-th and $l$-th canonical vectors---are indicator vectors of valid stable sets, thus
\[\phi_i(e_k) = 1 \wedge \phi_i(e_l) = 1 \quad \forall\, i.\]

We conclude the proof by noting that, since $e_k$ and $e_l$ only differ in the $k$-th and $l$-th coordinates, and $I_k \cap I_l = \emptyset$, 
\[\phi_i(e_k + e_l) = 1 \quad \forall\, i.\]

This is not possible, since both $k$ and $l$ cannot be part of a stable set of $G$ simultaneously.

 \end{proof}

Slightly abusing notation, we now define a $(\cdot)^+$ operator for
graphs (which is based on the $(\cdot)^+$ operator for polytopes) that
will justify why we can use the same procedure as in Theorem
\ref{theo:general} within the stable set family.

\begin{Def}
Let $G=(V,E)$ be a graph on $n$ vertices with $V=[n]$. We define $G^+$ as
\begin{enumerate}[(i)]
    \item $V(G^+) = [n+1]$.
    \item $E(G^+) = E(G) \cup \{(i,n+1) \, :\, \forall\, i \in [n]\}$
\end{enumerate}
\end{Def}

We are ready to formulate the following key lemma:

\begin{lemma} \label{stab-pyramid}
Let $G$ be a graph on $n$ vertices, and define $G^+$ as before. Then $conv(\text{STAB}(G^+))$
is a pyramid with base $conv(\text{STAB}(G))$ and apex $e_{n+1}$. Moreover,
\[ \text{STAB}(G^+) = \text{STAB}(G)^+.\]
\end{lemma}
\begin{proof}
  This follows directly since the stable sets of $G^+$ are either
  stable sets of $G$ or $\{n+1\}$ by construction. This is in
  correspondence to the definition of $\text{STAB}(G)^+$.
 \end{proof}

Using the lemma from above, we can retrace the proof of Theorem
\ref{theo:general} using as a starting point a family \(\{S_n\}_{n \in
\mathbb{N}}\) with $S_n = \text{STAB}(G_n)$ for some graph $G_n$ over $n$
nodes with \(n \in \mathbb{N}\). Note 
that, in addition, if we consider $G_1, G_2$ copies of a graph $G$ we
have
\[\text{STAB}(G_1 \cup G_2) = \text{STAB}(G_1) \times \text{STAB}(G_2).\]

All in all we obtain the following corollary:

\begin{cor}
  Given any sequence $\{w_n\}_{n=1}^\infty$ of natural numbers such that $w_n \leq n-1$
  for all \(n\), there exists a family of connected graphs
  $\{G'_n\}_{n=1}^\infty$ such that $tw(G_n') \leq w_n$ and
\[xc_{SDP} (\text{STAB}(G'_n)) \in \Omega \left(\frac{n}{w_n + 1} 2^{\Omega(w_n^{1/13})} \right).\]
\end{cor}
\begin{proof}
  The proof follows along the same lines as the proof of Theorem
  \ref{theo:general}. Our starting point is a result by
  \citet{lee2015lower}, where it is proven that for any $n$, there
  exists graph $G_n$ on $n$ vertices, such that
\[xc_{SDP}(STAB(G_n)) \geq 2^{\Omega(n^{1/13})}.\]
Employing Lemma \ref{stab-pyramid}, the operations $(\cdot)^+$ and
$\times$ over stable set polytopes correspond to operations performed
directly over graphs, thus the result follows; the connectedness
requirement follows from the separability argument stated before.
  \end{proof}

 We can also restrict ourselves to linear extension complexity to obtain
 sharper lower bounds. Following the exact same strategy (using
 Theorem \ref{theo:pyramids} instead of Theorem
 \ref{theo:pyramidsSDP}) and using a result of
 \citet{goos2016extension} that shows that there exist graphs $G$ on
 $n$ variables such that
\[xc(STAB(G_n)) \geq 2^{\Omega(n/\log n)} \]
we obtain the following corollary for the linear extension complexity case:

\begin{cor}\label{cor:SS}
Given any family $\{w_n\}_{n=1}^\infty$ such that $w_n \leq n-1$ for
all \(n\), there exists a family of graphs $\{G'_n\}_{n=1}^\infty$ such that $tw(G_n') \leq w_n$ and 
\[xc (\text{STAB}(G'_n)) \in \Omega \left(\frac{n}{w_n + 1} 2^{\Omega(w_n/\log w_n)} \right).\]
\end{cor}

\subsection{Average Extension Complexity of Stable Set Polytopes}

In Theorem \ref{thr:counting} (resp. Corollary \ref{cor:SS}) it is shown
that the upper bounds on the extension complexity in terms of
treewidth discussed in the introduction are essentially tight when we
consider $0/1$ polytopes (resp. stable set polytopes), i.e., there
exist polytopes that almost satisfy the bound. It is a natural
question to ask whether this fact holds with high probability: if we
sample a ``random'' $0/1$ (resp. stable set) polytope, is its
extension complexity close to the bound given by Theorem
\ref{thr:counting} (resp. Corollary \ref{cor:SS}) with high probability?

We show that the answer to the previous question is negative for
stable set polytopes in the classical Erd\"os-Renyi random graph
model. Here a graph $G(n,p)$ is the outcome of a random process that
starts from a graph on $n$ nodes without any edges and then adds each
potential edge independently with probability $p$. For a graph $G$, we denote by
$\alpha(G)$ the maximum size of its stable set. The average extension
complexity of stable set polytope has been studied in
\cite{BrFiPo}. However, we will only need the following basic fact,
that can be found in e.g. \cite[Lemma 11.2.1]{Diestel}.

\begin{lemma}\label{lem:diestel}
Let $n,r\geq 2$ and $G=G(n,p)$. Then $\mathbb{P}(\alpha(G)\geq r) \leq (ne^{-p(r-1)/2)})^r$.
\end{lemma}

For $\alpha(G)\geq 2$, standard enumeration arguments imply that
$xc(\text{STAB}(G))\leq n^{\alpha(G)}$. Hence, for $r \gg \log n/p$,
using Lemma \ref{lem:diestel} one has
\[\begin{array}{lll}\mathbb{P}(xc(\text{STAB} (G))\geq 2^{r \log n})=\mathbb{P}(xc(\text{STAB} (G))\geq n^{r})& \leq & \mathbb{P}(\alpha(G)\geq r) \\ & \leq & n^re^{-pr(r-1)/2}\\& = & 2^{r \log n}e^{-pr(r-1)/2}\xrightarrow[n\rightarrow +\infty]{}0.\end{array}\]
In the regime $p \gg c(n)\cdot \frac{\log^2 n}{n}$ with
$c(n)=\Omega(1)$, we have therefore
\[\lim_{n\rightarrow \infty}{\mathbb{P}(xc(\text{STAB} (G))\geq 2^{\frac{n}{c(n)\log n}})}=0.\]
On the other hand, random graphs in the same regime of $p$ have linear
treewidth with high probability, as shown in \cite[Theorem 2]{TwRandom}.

\begin{theorem}[\citet{TwRandom}]\label{thr:random-graphs}
  Let $p$ be as above and $G=G(n,p)$. Then
  \[\lim_{n\rightarrow +\infty} \mathbb{P}(tw(G)\geq (1-t)n)=1\] 
  for every constant $0<t<1$.
\end{theorem}

Therefore, for any $p$ in this regime, the corresponding bound for
$xc(\text{STAB} (G))$ given by Corollary \ref{cor:SS} is of the order
$\Omega(2^{\frac{n}{\log n}})$ with high probability. This means that
stable set polytopes like the ones constructed in Corollary
\ref{cor:SS} happen with very low probability, and thus the
corresponding treewidth-based upper bound is loose with high
probability.

Note that this also prevents us from using counting arguments (similar
to \cite{rothvoss2013some,briet2013existence}) to establish high
extension complexity as replacement for the construction used to establish Corollary
\ref{cor:SS} for stable sets.

\subsection{Lower bounds for a fixed intersection graph family}
Theorem \ref{thr:counting} holds for any valid ``target'' treewidth, and we assume complete freedom on the set that we can construct, as long
as its treewidth meets the target. A natural question is whether 
the same can be said for an arbitrary family of intersection graphs. In this section
we prove that a similar result can be obtained even if we fix the graph
family, and require the constructed sets to have the fixed family as intersection graph.
Since we are in a much more restricted setting, the result is weaker than Theorem \ref{thr:counting}, but it remains exponential in the treewidth parameter.

\subsubsection{Planar 2-SAT polytope with exponential extension complexity}

In \citet{avis2015generalization} it is shown that:

\begin{theorem}[\citet{avis2015generalization}]
  For every $n$ there exists a $2$-$SAT$ formula $\phi$ in $n$
  variables such that the satisfiability polytope of $\phi$ has
  extension complexity at least $2^{\Omega(\sqrt[4]{n})}$. Moreover,
  one can assume the graph induced by the $2$-SAT formula is planar.
\end{theorem}

Here, the satisfiability polytope is simply the convex hull of the 0/1
points that satisfy the boolean formula $\phi$. The result follows
from a stable set instance with linear extension complexity
$2^{\Omega(\sqrt{n})}$ from \cite{fiorini2012linear}, which is then used in
\cite{avis2015extension} to obtain a stable set instance on a
\emph{planar} graph with extension complexity
$2^{\Omega(\sqrt[4]{n})}$. The latter can be cast as the feasible set
of a 2-SAT formula derived from the underlying graph. 

One can
follow the same strategy along with the results regarding \emph{semidefinite} extension complexity of \citet{lee2015lower} to see that
there exist a family of 2-SAT formulas $\{\phi_n\}_{n=1}^\infty$ on
$n$ variables defined over \emph{planar} graphs, such that their
respective satisfiability polytopes have \emph{semidefinite} extension
complexity at least
\begin{equation} \label{lbstabpsd}
2^{\Omega(n^{1/26})}. 
\end{equation}

Using this observation, we can follow a similar strategy as in Section \ref{section2SATreduction} to prove the following lower bound.

\subsubsection{Lower bound result}

\begin{theorem} \label{arbitrarygraphs}
Let $\{G_k\}_{k=1}^{\infty}$ be an arbitrary family of graphs indexed by treewidth. There exists a sequence $\{(S_{n_k}, G_{n_k})\}_{k = 1}^\infty$, where
\begin{enumerate}[(i)]
\item $S_{n_k}$ is a 0/1 set.
\item $\{G_{n_k}\}_{k=1}^{\infty}$ is a subsequence of $\{G_k\}_{k=1}^{\infty}$.
\item $S_{n_k}$ admits a formulation with $G_{n_k}$ as intersection graph, and thus $tw(S_{n_k}) \leq n_k$.
\item $xc_{SDP}(S_{n_k}) \geq 2^{poly(n_k^{1/c})}$, for a universal
  constant $c$.
\end{enumerate}
\end{theorem}
\begin{proof}
Consider $\phi_n$ a 2-SAT formula on $n$ variables over a planar graph such that
\[xc_{SDP}(S'_n) \geq 2^{\Omega(n^{1/26})}\]
where $S'_n$ is the set of binary vectors that satisfy $\phi_n$. Let
$H_n$ be the planar graph on $n$ vertices associated to $\phi_n$ and fix
$k\in \mathbb{N}$. By Corollary
\ref{main-cor}, we know there exists a polynomial $\kappa(k)$ such that
$H_k$ is a minor of $G_{\kappa(k)}$. Following the proof in Section
\ref{section2SATreduction}, we can start from a formulation of
$S_n' \subseteq\{0,1\}^k$ and obtain an equivalent formulation in a
lifted space, which has $G_{\kappa(k)}$ as intersection graph. We call
the set of feasible solutions of this lifted formulation
$S'_{\kappa(k)}$.

Since the procedure generates equivalent formulations in an extended
space, one can easily see that
\[xc_{SDP}(S'_{k}) \leq xc_{SDP}(S'_{\kappa(k)})\]
consequently,
\[xc_{SDP}(S'_{\kappa(k)}) \geq 2^{\Omega(k^{1/26})}.\]
Defining $n_k = \kappa(k)$ the result follows. The fact that $c$ is a universal
constant is justified by the fact that $\kappa$ depends only on the number of vertices of $H_k$.
 \end{proof}

The reader might have noticed that one can state the result in Theorem
\ref{arbitrarygraphs} without the need of a subsequence $n_k$, since
one can augment the sequence by defining ``intermediate'' pairs
\begin{equation}\label{patch}
(S_{\kappa(k) + 1}, G_{\kappa(k) + 1}), \ldots, (S_{\kappa(k+1) -1}, G_{\kappa(k+1)-1}) 
\end{equation}
since $H_k$ is also a minor of $G_{k'}$ for $k' > \kappa(k)$. This
would not change the exponential lower bound, since there is only a
polynomial number of sets in \eqref{patch}.

\section*{Acknowledgements}
Research reported in this paper was partially supported by NSF CAREER
award CMMI-1452463 and by the Institute for Data Valorisation (IVADO).

\bibliographystyle{spbasic}      
\bibliography{pwidth}

\newpage

\end{document}